%% file: general_m.tex
\documentclass{article}
\usepackage[T1]{fontenc}
\usepackage[latin9]{inputenc}
\usepackage{geometry}
\usepackage{float}
\usepackage{ifthen}
\usepackage{amsbsy}
\usepackage{amssymb}
\usepackage{amsmath}
\usepackage{amsthm}
\usepackage{graphicx}
\usepackage{setspace}
\usepackage{esint}
\usepackage{comment}
\usepackage{mathtools}
\usepackage{xcolor}

\input{dylan-macros}

\usepackage{nicefrac}
\usepackage{subcaption}
\usepackage{bbm}
\usepackage{booktabs}

\makeatletter
\newtheorem*{rep@theorem}{\rep@title}
\newcommand{\newreptheorem}[2]{%
\newenvironment{rep#1}[1]{%
 \def\rep@title{#2 \ref{##1}}%
 \begin{rep@theorem}}%
 {\end{rep@theorem}}}
\makeatother

\usepackage{url}
\usepackage{hyperref}
\hypersetup{breaklinks}
\usepackage{cleveref}
\newtheorem{theorem}{Theorem}
\newreptheorem{theorem}{Theorem}
\newtheorem{corollary}[theorem]{Corollary}
\newreptheorem{corollary}{Corollary}
\newtheorem{lemma}[theorem]{Lemma}

\newtheorem{example}{Example}
\newtheorem{remark}{Remark}

\crefname{equation}{}{}
\crefname{proposition}{Proposition}{Propositions}
\crefname{appendix}{Appendix}{Appendices}
\usepackage{autonum}
\usepackage{authblk}

\usepackage{color-edits}
\addauthor{vs}{red}
\addauthor{rs}{blue}
\addauthor{vc}{purple}
\addauthor{wn}{green}

\newcommand{\kibitz}[2]{\ifnum\Comments=1{\color{#1}{#2}}\fi}

\newcommand{\E}{\mathbb{E}}

\newcommand{\R}{\mathbb{R}}

\newcommand{\mcX}{{\mathcal X}}

\newcommand{\mcG}{{\mathcal G}}

\newcommand{\ba}{\begin{array}}
\newcommand{\ea}{\end{array}}
\newcommand{\bs}{\begin{align}\begin{split}\nonumber}
\newcommand{\bsnumber}{\begin{align}\begin{split}}
\newcommand{\es}{\end{split}\end{align}}

\newcommand{\mcZ}{{\mathcal Z}}

\def\defeq{\triangleq} %
 
\def\balign#1\ealign{\begin{align}#1\end{align}}
\def\balignat#1\ealign{\begin{alignat}#1\end{alignat}}
\def\bitemize#1\eitemize{\begin{itemize}#1\end{itemize}}
\def\benumerate#1\eenumerate{\begin{enumerate}#1\end{enumerate}}
\newenvironment{talign}
 {\let\displaystyle\textstyle\csname align\endcsname}
 {\endalign}
\def\balignt#1\ealignt{\begin{talign}#1\end{talign}}%
\title{Debiased Machine Learning without Sample-Splitting for Stable Estimators}
\author[1]{Qizhao Chen}
\author[2]{Vasilis Syrgkanis}
\author[3]{Morgane Austern}
\affil[1,3]{Department of Statistics, Harvard University}
\affil[2]{MS\&E Department, Stanford University}
\date{}

\begin{document}

\maketitle

\begin{abstract}
    Estimation and inference on causal parameters is typically reduced to a generalized method of moments problem, which involves auxiliary functions that correspond to solutions to a regression or classification problem. Recent line of work on debiased machine learning shows how one can use generic machine learning estimators for these auxiliary problems, while maintaining asymptotic normality and root-$n$ consistency of the target parameter of interest, while only requiring mean-squared-error guarantees from the auxiliary estimation algorithms. The literature typically requires that these auxiliary problems are fitted on a separate sample or in a cross-fitting manner. We show that when these auxiliary estimation algorithms satisfy natural leave-one-out stability properties, then sample splitting is not required. This allows for sample re-use, which can be beneficial in moderately sized sample regimes. For instance, we show that the stability properties that we propose are satisfied for ensemble bagged estimators, built via sub-sampling without replacement, a popular technique in machine learning practice.
\end{abstract}

\section{Introduction}

A large variety of problems in causal inference and more generally semi-parametric inference can be framed as finding a solution to a moment condition:
\begin{align}
M(\theta, g) ~\defeq~& \E_Z[m(Z; \theta, g)] & 
    M(\theta_0, g_0) ~=~& 0
\end{align}
where $Z\in \mcZ$ is a vector of random variables that, apart from the target parameter $\theta_0\in \Theta\subset\mathbb{R}^p$ of interest, also depends on unknown nuisance functions $g_0\in \mcG$, which need to be estimated in a flexible manner from the data. This framework has a long history in the literature on semi-parametric inference \cite{levit1976efficiency,hasminskii1979nonparametric,ibragimov1981statistical,pfanzagl1982lecture,klaassen1987consistent,robinson1988root,van1991differentiable,bickel1993efficient,newey1994asymptotic,robins1995semiparametric,vaart,bickel1988estimating,newey1998undersmoothing,ai2003efficient,newey2004twicing,ai2007estimation,tsiatis2007semiparametric,kosorok2007introduction,ai2012semiparametric}, which  analyzes the following two-stage estimation process with sample re-use, when having access to $n$ iid samples $\{Z_1, \ldots, Z_n\}$ and as $n$ grows, treating the target parameter dimension $p$ as a constant:
\begin{enumerate}
    \item Obtain an estimate $\hat{g}\in \mcG$ of the nuisance function $g_0$ based on all the samples.
    \item Return any estimate $\hat{\theta}\in \Theta$ that satisfies:
    \begin{align}\label{eqn:crossfit-estimate}
    M_{n}(\hat{\theta}, \hat{g}) ~=~& o_p\left(n^{-1/2}\right) & & \text{ with } & M_n(\theta, g) ~\defeq~& \frac{1}{n} \sum_{i=1}^n m(Z_i; \theta, g) .
    \end{align}
\end{enumerate}
Given that the estimation of function $\hat{g}$ is a complex non-parametric problem, it will typically only satisfy slower than parametric error rates. Moreover, in high dimensional settings, when machine learning techniques are used to estimate $\hat{g}$, then the regularization bias of $\hat{g}$ will propagate to the final estimate $\hat{\theta}$, leading to non-regular estimates and the inability to construct confidence intervals. 

The literature on efficient semi-parametric estimation provides conditions on the moment function so that the influence of the estimation error of the nuisance function $\hat{g}$ is of second order importance and does not alter the distributional properties of the second stage
\cite{hasminskii1979nonparametric,bickel1988estimating,zhang2014confidence,belloni2011inference,belloni2014inference,belloni2014uniform,belloni2014pivotal,javanmard2014confidence,javanmard2014hypothesis,javanmard2018debiasing,van2014asymptotically,ning2017general,chernozhukov2015valid,neykov2018unified,ren2015asymptotic,jankova2015confidence,jankova2016confidence,jankova2018semiparametric,bradic2017uniform,zhu2017breaking,zhu2018linear}. 
This approach dates back to the classical work on doubly robust estimation and targeted maximum likelihood \cite{robins1995analysis,robins1995semiparametric,van2006targeted,van2011targeted,luedtke2016statistical,toth2016tmle}
as well as the more recent work on locally robust or Neyman orthogonal conditions on the moment function \cite{neyman1959,neyman1979c,chernozhukov2016locally,belloni2017program,chernozhukov2018double}, typically referred to as double or debiased machine learning. At a high-level, the earlier literature on semi-parametric inference shows that if the moment function satisfies some form of robustness to nuisance perturbations and, importantly, as long as the space $\mcG$ used in the first stage estimation is a relatively simple function class, in terms of statistical complexity, typically referred to as a Donsker function class, then the second stage estimate is root-$n$ consistent and asymptotically normal. Hence, one can easily construct confidence intervals for the target parameter of interest.

Crucially the recent literature on debiased machine learning alters the standard two-stage estimation algorithm to introduce the idea of sample-splitting, \cite{bickel1982adaptive,zheng2010asymptotic}. In particular, instead of estimating $\hat{g}$ on all the samples, the recent work on debiased machine learning \cite{chernozhukov2018double,chernozhukov2016locally,chernozhukov2021simple} estimates $\hat{g}$ on a separate sample, or for better sample efficiency invokes ``cross-fitting,'' where we train a nuisance model on half the data and evaluate it in the second stage on the other half and vice versa. Sample splitting avoids the Donsker conditions that were prevalent in the classic semi-parametric inference literature and only requires a root-mean-squared-error (RMSE) guarantee of the estimate $\hat{g}$ of $o_p(n^{-1/4})$.

However, sample splitting or cross-fitting still leads to poorer sample usage, as we can lose half of our data when training complex non-parametric or machine learning models, which can be problematic in small and moderate sample regimes. Our main result is to show that root-$n$ consistency and asymptotic normality of the standard algorithm, without sample splitting, can be achieved without the Donsker property, but solely if one assumes that the first stage estimation algorithm is $o(n^{-1/2})$ leave-one-out stable, a relatively widely studied property in the statistical machine learning and generalization theory literature \cite{bousquet2002stability,Kale11cross-validationand,elisseeff2003leave,Hardt2016}. As a leading example we show that our stability conditions are satisfied by bagging estimators formed with sub-sampling without replacement. 

Recent prior work of \cite{chernozhukov2020adversarial} also considered asymptotic normality based on stability conditions, but as we expand in the main text, their requirement on the stability property is much harsher than the one we derive here. Moreover, \cite{chernozhukov2020adversarial} analyzes only a special case of the class of moment problems that we consider here. For instance, in our leading example of bagging estimators, the prior result of \cite{chernozhukov2020adversarial} would require that the bias of the base estimator decays faster than $1/n$, where $n$ is the sample size, which is typically not the case. In contrast, our stability condition does not impose any explicit assumption on the bias and solely requires that the sub-sample size $m$ is $o(\sqrt{n})$. 

To simplify the regularity assumptions required for asymptotic normality, we focus on the case where $m(Z;\theta,g)$ is linear in $\theta$, i.e.
\begin{align}
    m(Z;\theta,g) = a(Z;g)\, \theta + \nu(Z;g)
\end{align}
where $a(Z;g)\in \R^{p\times p}$ is a $p\times p$ matrix and $\nu(Z;g)\in \R^p$ is a $p$-vector, and we denote with:
\begin{align}
    A(g) :=~& \E_Z[a(Z;g)] & 
    A_n(g) :=~&  \E_n[a(Z;g)]\\
    V(g) :=~& \E_Z[\nu(Z;g)] & 
    V_n(g) :=~& \E_n[\nu(Z;g)].
\end{align}
Many of the leading examples in semi-parametric problems that arise in causal inference correspond to linear moment problems. We present below a representative set of problems that are widely used in the practice of causal inference. 

\begin{example}[Partially Linear Treatment Effect \cite{robinson1988root}] \label{sleep2} If one assumes that the outcome of interest $Y$ is linear in the treatment, i.e. $Y=\theta_0'T + f_0(X) + \epsilon$, with $\E[\epsilon\mid T, X]=0$, then estimating the treatment effect $\theta_0$ boils down to solving the following linear moment:
\begin{align}
    m(Z;\theta,g) = (Y - q(X) - \theta'(T-p(X)))\, (T - p(X))
\end{align}
where the corresponding true values of $(q,p)$ are $q_0(X)=\E[Y|X]$, $p_0(X)=\E[T|X]$.
\end{example}

\begin{example}[Partially Linear IV \cite{chernozhukov2018double}] \label{sleep3} If one assumes that the outcome of interest $Y$ is linear in the treatment, i.e. $Y=\theta_0'T + f_0(X) + \epsilon$, but the treatment is endogenous (i.e. there are unobserved confounders) and one has access to a random variable $Z$, that is referred to as an instrument, which correlates with the treatment but is un-correlated with the residual in the outcome equation, i.e. satisfies that $\E[\epsilon\mid Z, X]=0$, then estimating $\theta_0$ boils down to solving the following linear moment:
\begin{align}
    m(Z;\theta,g) = (Y - q(X) - \theta'(T-p(X)))\, (Z - r(X))
\end{align}
where the true values of $(q,p,r)$ are $q_0(X)=\E[Y|X]$, $p_0(X)=\E[T|X]$, $r_0(X) = \E[Z|X]$.
\end{example}

\begin{example}[Average linear functionals of regression functions]\label{sleep1} Consider a class of moment functions of the form:
\begin{align}
    m(Z;\theta,g) = \theta - m_{b}(Z;q) - \mu(T,X)\, (Y - q(T,X))
\end{align}
where $g=(q,\mu)$, $m_{b}$ is a linear functional of $q$ and the corresponding true value $q$ is a regression function $q_0(T, X)=\E[Y\mid T, X]$ and $\mu_0$ is Riesz representer of the functional $\E[m_b(Z;q)]$ (see \cite{chernozhukov2018global} for examples). For instance, in the case of a binary treatment $T$, where we have that $Y = q(T, X) + \epsilon$ and $\E[\epsilon\mid T, X]=0$, then the average treatment effect $\theta_0=\E[q(1,X) - q(0,X)]$ is identified by a moment of the latter type, with:
\begin{align}
    m_b(Z;q) = q(1, X) - q(0, X) \tag{Average Treatment Effect (ATE)}
\end{align}
While if we have a target treatment policy $\pi: \mcX \to \{0,1\}$, and we want to estimate its average value $\theta_0$, we can identify do so with a moment of the aforementioned type, with:
\begin{align}
    m_b(Z;q) = \pi(X)\, (q(1, X) - q(0, X)) \tag{Average Policy Effect}
\end{align}
\end{example}

For completeness, we also include in Appendix~\ref{sec: nonlinear} an extension of our results to nonlinear moment problems.

\section{Asymptotic Normality without Sample Splitting}

We start by providing an asymptotic normality theorem for semi-parametric moment estimators without sample splitting and where the moment satisfies the well-studied property of Neyman orthogonality. Our theorem requires four main conditions: i) root-mean-squared-error (RMSE) rates for the nuisance function estimates of $o_p(n^{-1/4})$, ii) Neyman orthogonality of the moment with respect to the nuisances, iii) second-order smoothness of the moment with respect to the nuisance functions and, iv) stochastic equicontinuity of the Jacobian and the offset part of the linear moment function as the nuisance estimate $\hat{g}$ converges to $g_0$. For a vector $x\in \R^p$ we denote with $\|x\|_2$ the $\ell_2$ norm and for a matrix $X\in \R^{p\times p}$ we denote with $\|X\|_{op}$ the operator norm with respect to the $\ell_2$ norm.

\begin{theorem}\label{thm:main}
Suppose that the nuisance estimate $\hat{g}\in \mcG$ satisfies:
\begin{equation}\label{cons}
\|\hat{g} - g_0\|^2_2~\defeq~ \E_X\left[\|\hat{g}(X) - g_0(X)\|_2^2\right] =o_p\left(n^{-1/2}\right). \tag{Consistency Rate}
\end{equation}
Suppose that the moment satisfies the Neyman orthogonality condition: for all $g\in \mcG$
\begin{align}
    D_g M(\theta_0, g_0)[g-g_0] ~\defeq~ \frac{\partial}{\partial t} M(\theta_0, g_0 + t\, (g-g_0))\big|_{t=0} ~=~& 0 \tag{Neyman Orthogonality}
\end{align}
and a second-order smoothness condition: for all $g\in \mcG$
\begin{align}
    D_{gg} M(\theta_0, g_0)[g-g_0] ~\defeq~ \frac{\partial^2}{\partial t^2} M(\theta_0, g_0 + t\, (g-g_0))\big|_{t=0} ~=~& O\left(\|g-g_0\|_2^2\right) \tag{Smoothness}
\end{align}
and that the moment $m$ satisfy the stochastic equicontinuity conditions:
\begin{align}\label{eqn:equicont}
\begin{aligned}
    \sqrt{n} \left\|A(\hat{g}) - A(g_0) - (A_n(\hat{g}) - A_n(g_0))\right\|_{op} =~& o_p(1)\\
    \sqrt{n} \left\|V(\hat{g}) - V(g_0) - (V_n(\hat{g}) - V_n(g_0))\right\|_2 =~& o_p(1).
\end{aligned} \tag{Stochastic Equicontinuity}
\end{align}
Assume that $A(g_0)^{-1}$ exists and that for any $g,g'\in \mcG$:
\begin{equation}
\|A(g) - A(g')\|_{op}=O\left(\|g-g'\|_2\right).    
\end{equation}
Moreover, assume that for any $i,j\in [p]\times[p]$, the random variable $a_{i,j}(Z;g_0)$ has bounded variance. Then $\hat{\theta}$ is asymptotically normal:
\begin{equation}
    \sqrt{n} \left(\hat{\theta} - \theta_0\right) \xrightarrow{n\rightarrow \infty, d} N\left(0, A(g_0)^{-1} \E\left[m(Z; \theta_0, g_0)\, m(Z; \theta_0, g_0)^\top\right] A(g_0)^{-1}\right).
\end{equation}
\end{theorem}

The first three conditions are standard assumptions in the literature on debiased machine learning. The final condition (stochastic equicontinuity) is exactly where sample splitting comes very handy in the literature. To illustrate the reason why sample splitting helps with the stochastic equicontinuity condition, let us consider the first part of the condition (the reasoning is analogous for the second part). It asks that the difference of two centered empirical processes, namely $A_n(\hat{g})-A(\hat{g})$ and $A_n(g_0) - A(g_0)$, goes to zero faster than $n^{-1/2}$. We expect each empirical process to go down to zero at exactly $n^{-1/2}$ and so this condition asks, since $\hat{g}$ converges to $g_0$ is the empirical process continuous in its argument and for that reason does the difference converge to zero faster than each individual component. If the estimate $\hat{g}$ was fitted on a separate sample, then conditional on $\hat{g}$, we have that each element $t=(i,j)\in [p]\times[p]$ of $A_n(\hat{g}) - A_n(g_0)$ is an empirical average of iid random variables with mean $A(\hat{g})-A(g_0)$. Thus a simple Bernstein inequality would show that the difference of the two empirical processes would converge to zero at the order of:
\begin{align}
    O_p\left(\sqrt{\frac{\E[(a_t(Z;\hat{g}) - a_t(Z;g_0))^2]}{n}} + \frac{1}{n}\right) = O_p\left(\sqrt{\frac{\|\hat{g}-g_0\|_2^2}{n}} + \frac{1}{n}\right) = o_p(n^{-1/2})
\end{align}
where we also invoked a mean-squared-continuity property of $a_t(Z;g)$ and the fact that $\|\hat{g}-g_0\|_2=o_p(1)$. Thus, with sample splitting, no further constraint is required from $\hat{g}$, other than a convergence rate on $\|\hat{g}-g_0\|_2$. In fact, as was noted in recent work of \cite{chernozhukov2021simple}, in the above step it suffices to assume that $\E[(a_t(Z;\hat{g}) - a_t(Z;g_0))^2] = O\left(\|\hat{g}-g_0\|_2^q\right)$ for any $q<\infty$, which is a much weaker mean-squared-continuity assumption, and the property would still hold, since $\|\hat{g}-g_0\|_2^{q/2} n^{-1/2}=o(n^{-1/2})$, whenever $\|\hat{g}-g_0\|_2=o_p(1)$.
 
Without sample splitting, note that $\hat{g}$ is now correlated with the samples in the empirical averages and hence $A_n(\hat{g}) - A_n(g_0)$ is no longer an average of i.i.d. random variables. Typical approaches would try to prove a uniform stochastic equicontinuity property over the function space $\mcG$, typically referred to as a Donsker property of the function space $\mcG$. In particular, if we could show that with high probability:
\begin{align}
    \forall g\in \mcG: \left\|A({g}) - A(g_0) - (A_n({g}) - A_n(g_0))\right\|_{op} = O\left(\delta_n \|g-g_0\|_2 + \delta_n^2\right) =
    O\left(\delta_n^2 +  \|g-g_0\|_2^2\right)
\end{align}
then the above property would also hold for $\hat{g}$. Subsequently, since we know that $\|\hat{g}-g_0\|_2^2=o_p(n^{-1/2})$, by our convergence rate assumptions on $\hat{g}$, then it would suffice that $\delta_n^2 = o(n^{-1/2})$. Such localized concentration inequalities have been known to hold for Donsker classes, which are typically defined via entropy integrals, and more recently it was also noted that such inequalities are satisfied with $\delta_n$ being the critical radius of the space $\mcG$, defined via localized Rademacher complexities.

However, the latter approach is conservative as it requires a uniform control over the function space $\mcG$ and does not utilize at all the properties of the estimation algorithm itself. In particular, as we will show in the next section, the main result of our work is that this stochastic equicontinuity condition follows from $o(n^{-1/2})$ leave-one-out stability conditions on our estimation algorithm, which are typical in the machine learning literature and in the excess risk and generalization bounds literature.

\section{Stochastic Equicontinuity via Stability}

We will show that the Condition~\eqref{eqn:equicont} is satisfied, whenever the estimate $\hat{g}$ satisfies leave-one-out stability properties and the moment satisfies the weak mean-squared-continuity property of \cite{chernozhukov2021simple}. We start by some preliminary definitions required to state our stability conditions. Define $Z^{(-l)}$ as the data $Z_1,\ldots,Z_n$ with the $l$-th data point $Z_l$ replaced with an independent copy $\tilde{Z}_l$.
Define $Z^{(-l_1,-l_2)}$ as the data $Z_1,\ldots,Z_n$ with both the $l_1$-th and the $l_2$-th data points $Z_{l_1},Z_{l_2}$ replaced with independent copies $\tilde{Z}_{l_1},\tilde{Z}_{l_2}$. Define similarly for $Z^{(-l_1,-l_2,-l_3)}$ and so on.
Let $\hat g^{(-l)}$ be the estimator trained on $Z^{(-l)}$ instead of $Z_1,\ldots,Z_n.$
Similarly, let $\hat g^{(-l_1,-l_2)}$ be trained on $Z^{(-l_1,-l_2)}$ and so on. Moreover, we will always denote with $Z$ a fresh random variable drawn from the distribution of the samples, but which is not part of any training sample. For any random variable $X$, we denote with $\|X\|_1:=\E[|X|]$, with $\|X\|_2:=\sqrt{\E[X^2]}$, and with $\|X\|_p:=\left(\E[|X|^p]\right)^{1/p}$ for any $p\ge 1$ in general. 

\begin{lemma}[Main Lemma]\label{lem:stoc-eq-jac}
If the estimation algorithm satisfies the stability conditions: for all $i,j\in [p]$
\begin{align}
\max_{l\in[n]}\Big\|a_{i,j}(Z_l,\hat g)-a_{i,j}(Z_l,\hat g^{(-l)})\Big\|_1=~& o(n^{-1/2}) &
\max_{l\in[n]}\Big\|a_{i,j}(Z,\hat g)-a_{i,j}(Z,\hat g^{(-l)})\Big\|_2=~& o(n^{-1/2})\\
\max_{l\in[n]}\Big\|\nu_{i}(Z_l,\hat g)-\nu_{i}(Z_l,\hat g^{(-l)})\Big\|_1=~& o(n^{-1/2}) &
\max_{l\in[n]}\Big\|\nu_{i}(Z,\hat g)-\nu_{i}(Z,\hat g^{(-l)})\Big\|_2=~& o(n^{-1/2})
\end{align}
and the moment satisfies the mean-squared-continuity condition: for all $i,j\in [p]$
\begin{align}
    \forall g,g': \E[(a_{i,j}(Z;g) - a_{i,j}(Z;g'))^2]\leq~& L \|g-g'\|_2^q & \E[(\nu_{i}(Z;g) - \nu_{i}(Z;g'))^2]\leq L \|g-g'\|_2^q
\end{align}
for some $q<\infty$ and some $L>0$, then the Condition~\eqref{eqn:equicont} is satisfied.
\end{lemma}

\begin{remark}
We show in \cref{bla} that the stability conditions are tight. We present a counter example for which $\Big\|\nu_{i}(Z,\hat g)-\nu_{i}(Z,\hat g^{(-l)})\Big\|_2$ is exactly of order $n^{-1/2}$ and for which the Stochastic Equicontinuity condition is not satisfied. 
\end{remark}

\begin{remark}
We note that prior work of \cite{chernozhukov2020adversarial} that established asymptotic normality without sample splitting via stability, required significantly stronger conditions than what we invoke here. In particular, if we let $\beta_n$ be the stability of the estimator $\hat{g}$ as measured by the quantities in Lemma~\ref{lem:stoc-eq-jac}, then the prior work of \cite{chernozhukov2020adversarial}, would require that $n \beta_n \|\hat{g}-g_0\|_2\to 0$. If we only know that $\|\hat{g}-g_0\|_2=o(n^{-1/4})$, then the above would require $\beta_n = o(n^{-3/4})$, which is much slower than $o(n^{-1/2})$. Moreover, for bagged kernel estimators that we analyze in section~\ref{sec:baggapp}, the prior work would require that if we use bags of size $m$, then the bias of the base estimator with $m$ samples, denoted as $\text{bias}(m)$ satisfies that $m \text{bias}(m)\to 0$. This would rarely be satisfied and in prior work, the only concrete case that was given was forest estimators with binary variables under strong sparsity conditions, in which case the bias decays exponentially with the sample size. For more general estimators, we expect $\text{bias}(m)=1/m^{\alpha}$, for some $\alpha$. For such settings, our work still applies and, as we show in section~\ref{sec:baggapp}, gives results for bagged $1$-nearest neighbor estimation algorithms, which do not satisfy any entropy or critical radius bound, but are stable. The key innovation that enables our improved results is a ``double centering'' approach that derives intuition from techniques invoked in the analysis of cross-validation via stability and the proof of the Efron-Stein inequality\cite{boucheron2003concentration}. This idea has already been used in the study of the cross validated risk \cite{austern2020asymptotics,bayle2020cross}.\end{remark}

\begin{proof}[Proof of Main Lemma]
We will show the first part of the lemma, i.e. that if for all $i,j\in [p]$
\begin{align}
\max_{l\in[n]}\Big\|a_{i,j}(Z_l,\hat g)-a_{i,j}(Z_l,\hat g^{(-l)})\Big\|_1=~& o(n^{-1/2}) &
\max_{l\in[n]}\Big\|a_{i,j}(Z,\hat g)-a_{i,j}(Z,\hat g^{(-l)})\Big\|_2=~& o(n^{-1/2})
\end{align}
then 
\begin{align}
  \sqrt{n} \left\|A(\hat{g}) - A(g_0) - (A_n(\hat{g}) - A_n(g_0))\right\|_{op} =~& o_p(1)
\end{align}
The analogous statement for $\nu$ and $V$, follows in an identical manner.

Since $A(g)$ and $A_n(g)$ are $p\times p$ matrices and $p=O(1)$, it suffices to show the above property for every element $(i,j)\in [p]\times [p]$, i.e. that
$$\sqrt{n} \left|A_{i,j}(\hat{g}) - A_{i,j}(g_0) - (A_{n,i,j}(\hat{g}) - A_{n,i,j}(g_0))\right|=o_p(1).$$
For this it suffices to show that:
$$J_n:=\sqrt{n} \left\|A_{i,j}(\hat{g}) - A_{i,j}(g_0) - (A_{n,i,j}(\hat{g}) - A_{n,i,j}(g_0))\right\|_1=o(1).$$

\textbf{In the remainder of the proof we look at a particular $(i,j)$ and hence for simplicity we overload notation and we let $a:=a_{i,j}$ and $A:=A_{i,j}$.}

By triangle inequality and monotonicity of $L^p$ norms we have
\begin{align}
    J_n =~& \left\|\frac{1}{\sqrt{n}}\sum_{l=1}^n \left\{a(Z_l,\hat g)-A(\hat g)-\big[a(Z_l,g_0)-A(g_0)\big]\right\}\right\|_1
   \\
    \le~& \left\|\frac{1}{\sqrt{n}}\sum_{l=1}^n \left\{ a(Z_l,\hat g) - a(Z_l,\hat g^{(-l)})\right\}\right\|_1
   +\sqrt{n}\max_{l\in[n]}\left\|A(\hat g)-A(\hat g^{(-l)})\right\|_1\\
   ~& +\left\|\frac{1}{\sqrt{n}}\sum_{l=1}^n\left\{a(Z_l,\hat g^{(-l)})-a(Z_l,g_0)-\big(A(\hat g^{(-l)})-A(g_0)\big)\right\}\right\|_2.
\end{align}

To ease notations, we denote
\begin{align}
J_{1,n}:=~&\left\|\frac{1}{\sqrt{n}}\sum_{l=1}^n \left\{ a(Z_l,\hat g) - a(Z_l,\hat g^{(-l)})\right\}\right\|_1,\\
J_{2,n}:=~& \sqrt{n}\max_{l\in[n]}\left\|A(\hat g)-A(\hat g^{(-l)})\right\|_1,\\
J_{3,n}:=~& \left\|\frac{1}{\sqrt{n}}\sum_l\left\{a(Z_l,\hat g^{(-l)})-a(Z_l,g_0)-\big(A(\hat g^{(-l)})-A(g_0)\big)\right\}\right\|_2.
\end{align}

Now we have that by triangle inequality
\begin{align}
J_{1,n}\le \frac{1}{\sqrt{n}}\sum_l\Big\|a(Z_l,\hat g)-a(Z_l,\hat g^{(-l)})\Big\|_1 \le\sqrt{n}\max_{\ell\in[n]} \Big\|a(Z_l,\hat g)-a(Z_l,\hat g^{(-l)})\Big\|_1=o(1).
\end{align}

Similarly we can handle $J_{2,n}$: 
\begin{align}
J_{2,n}\le~& \sqrt{n}\max_{l\in [n]} \Big\|a(Z,\hat g)-a(Z,\hat g^{(-l)})\Big\|_1 \\
\leq~&  \sqrt{n}\max_{l \in [n]}\Big\|a(Z,\hat g)-a(Z,\hat g^{(-l)})\Big\|_2=o(1).
\end{align}

We now aim to show that $J_{3,n}=o(1)$. Write for simplicity
$$K_l:=a(Z_l,\hat g^{(-l)})-a(Z_l,g_0)-\big(A(\hat g^{(-l)})-A(g_0)\big)$$
Now by expanding the square we obtain
\begin{align}
    J_{3,n}^2=\E\left[\left( \frac{1}{\sqrt{n}}\sum_{l=1}^n K_l\right)^2\right]
    &\leq  \max_{l\in [n]} \E\left[K_l^2\right] + (n-1) \max_{l_1\neq l_2\in [n]} \E[K_{l_1} K_{l_2}]
\end{align}

To bound the first term, we have by mean-squared continuity:
\begin{align}
    \E[K_l^2] =~&\E\left[\left(a(Z_l,\hat g^{(-l)})-a(Z_l,g_0)-\big(A(\hat g^{(-l)})-A(g_0)\big)\right)^2\right]\\
    \le~& 2\E\left[\left(a(Z_l,\hat g^{(-l)})-a(Z_l,g_0)\right)^2\right]+2\E\left[\left(A(\hat g^{(-l)})-A(g_0)\right)^2\right]\\
    =~& 2\E\left[\E\left[\left(a(Z_l,\hat g^{(-l)})-a(Z_l,g_0)\right)^2\Big| Z^{(-l)}\right]\right]+2\E\left[\E\left[\left(A(\hat g^{(-l)})-A(g_0)\right)^2\Big|Z^{(-l)} \right]\right]\\
    \le~& 2L\cdot\E\left[\left\|\hat g^{(-l)}-g_0\right\|_2^q\right]+2L\cdot\E\left[\left\|\hat g^{(-l)}-g_0\right\|_2^q\right]
    = 4L\cdot\mathbb{E}\left[\left\|\hat g-g_0\right\|_2^q\right]=o(1).
\end{align}
where we invoked the property that $\left\|\hat g-g_0\right\|_2=o_p(1)$, and the second to last equality exploited the tower law. 
Thus $\max_{l\in [n]} \E\left[K_l^2\right]=o(1)$.


\paragraph{Double centering.} We now bound the term, $(n-1) \max_{l_1\neq l_2\in [n]} \E[K_{l_1} K_{l_2}]$. Define, for simplicity, for $l_1\neq l_2$
$$K_{l_1}^{(l_2)}:=a(Z_{l_1},\hat g^{(-l_1,-l_2)})-a(Z_{l_1},g_0)-\big(A(\hat g^{(-l_1,-l_2)})-A(g_0)\big).$$
Note that $K_{l_1}^{(l_2)}$ does not depend on the $l_2$-th data point and is only a function of $Z^{(-l_2)},\tilde{Z}_{l_1}.$\\
Moreover, by the definition of $A$, noting that $Z_l$ is independent of $Z^{(-l)}$ and $Z_{l_1}$ is independent of $Z^{(-l_1,-l_2)}$: 
\begin{align}
    \E\left[K_l\Big|Z^{(-l)}\right]
    =~& \E\left[a(Z_l,\hat g^{(-l)})-A(\hat g^{(-l)})\Big|Z^{(-l)}\right]-\mathbb{E}\left[a(Z_l,g_0)-A(g_0)\Big|Z^{(-l)}\right]\\
    =~& \E\left[a(Z_l,\hat g^{(-l)})-A(\hat g^{(-l)})\Big|Z^{(-l)}\right]-\mathbb{E}\left[a(Z_l,g_0)-A(g_0)\right]=0
\end{align}
and
\begin{align}
    &\mathbb{E}\left[K_{l_1}^{(l_2)}\Big| Z^{(-l_1,-l_2)}\right]\\
    &=\mathbb{E}\left[a(Z_{l_1},\hat g^{(-l_1,-l_2)})- A(\hat g^{(-l_1,-l_2)})\Big| Z^{(-l_1,-l_2)}\right]+\mathbb{E}\left[a(Z_{l_1},g_0)-A(g_0)\Big| Z^{(-l_1,-l_2)}\right]\\
    &=\mathbb{E}\left[a(Z_{l_1},\hat g^{(-l_1,-l_2)})- A(\hat g^{(-l_1,-l_2)})\Big| Z^{(-l_1,-l_2)}\right]+\mathbb{E}\left[a(Z_{l_1},g_0)-A(g_0)\right] = 0
\end{align}

For simplicity, we show that $(n-1) \E[K_1 K_2]=o(1)$, i.e. we show that the second term vanishes for $l_1=1$ and $l_2=2$. The same exact arguments generalize to arbitrary $l_1, l_2$. We begin by writing:
\begin{align}
    (n-1) \E[K_{1} K_{2}]
    =(n-1)\mathbb{E}\left[\left(K_{1}-K_{1}^{(2)}\right)K_{2}\right],
\end{align}
since by tower law:
\begin{align}
        \mathbb{E}\left[K_{1}^{(2)}K_{2}\right]=\mathbb{E}\left[\mathbb{E}\left[K_{2}\Big|Z^{(-2)},\tilde{Z}_{1}\right]K_{1}^{(2)}\right]=\mathbb{E}\left[\mathbb{E}\left[K_2\Big|Z^{(-2)}\right]K_1^{(2)}\right]=0
\end{align}

Similarly by conditioning on $Z^{(-1)},\tilde{Z}_2$ and using tower law, we can show that
\begin{align}
    \E\left[\left(K_1-K_1^{(2)}\right)K_2^{(1)}\right]
    =~& \mathbb{E}\left[\mathbb{E}\left[K_1-K_1^{(2)}\Big| Z^{(-1)},\tilde{Z}_2\right]K_2^{(1)}\right]\\
    =~& \mathbb{E}\left[\left\{\mathbb{E}\left[K_1\Big| Z^{(-1)},\tilde{Z}_2\right]-\mathbb{E}\left[K_1^{(2)}\Big| Z^{(-1)},\tilde{Z}_2\right]\right\}K_2^{(1)}\right]\\
    =~& \mathbb{E}\left[\left\{\mathbb{E}\left[K_1\Big| Z^{(-1)}\right]-\mathbb{E}\left[K_1^{(2)}\Big| Z^{(-1,-2)}\right]\right\}K_2^{(1)}\right]=0.
\end{align}

Hence, we have
\begin{align}
(n-1)\mathbb{E}[K_1K_2]=(n-1)\mathbb{E}\left[\left(K_1-K_1^{(2)}\right)K_2\right]=(n-1)\mathbb{E}\left[\left(K_1-K_1^{(2)}\right)\left(K_2-K_2^{(1)}\right)\right]
\end{align} 
With identical arguments, the same equality holds for any indices $l_1, l_2$. By Cauchy-Schwarz:
\begin{align}
    \max_{l_1\neq l_2} (n-1)\mathbb{E}[K_{l_1} K_{l_2}] =~&  \max_{l_1\neq l_2} (n-1)\mathbb{E}\left[\left(K_{l_1}-K_{l_1}^{(l_2)}\right)\left(K_{l_2}-K_{l_2}^{(l_1)}\right)\right]\\
\leq~& \max_{l_1\neq l_2} (n-1) \left\|K_{l_1}-K_{l_1}^{(l_2)}\right\|_2 \left\|K_{l_2}-K_{l_2}^{(l_1)}\right\|_2
\end{align}
Thus for $\max_{l_1\neq l_2} (n-1)\mathbb{E}[K_{l_1} K_{l_2}]=o(1)$ it suffices that:
$\max_{l_1\neq l_2} \left\|K_{l_1}-K_{l_1}^{(l_2)}\right\|_2 = o(n^{-1/2})$.
Expanding the definitions $K_{l_1}$ and $K_{l_1}^{(l_2)}$, the above simplifies to:
\begin{align}
\left\|K_{l_1}-K_{l_1}^{(l_2)}\right\|_2 = \left\|a(Z_{l_1},\hat g^{(-l_1)})-A(\hat g^{(-l_1)})-\left(a(Z_{l_1},\hat g^{(-l_1,-l_2)})-A(\hat g^{(-l_1,-l_2)})\right)\right\|_2.
\end{align}
The latter is upper bounded by a triangle inequality and a Jensen's inequality by:
\begin{align}
    \left\|K_{l_1}-K_{l_1}^{(l_2)}\right\|_2 \leq 2 \left\|a(Z_{l_1},\hat g^{(-l_1)}) - a(Z_{l_1},\hat g^{(-l_1,-l_2)})\right\|_2.
\end{align}
If we denote with $Z$ a fresh random sample not part of the training sets, then since $Z_{l_1}$ is not part of $Z^{(-l_1)}$, we have:
\begin{align}
    \left\|a(Z_{l_1},\hat g^{(-l_1)}) - a(Z_{l_1},\hat g^{(-l_1,-l_2)})\right\|_2 =~& \left\|a(Z,\hat g^{(-l_1)}) - a(Z,\hat g^{(-l_1,-l_2)})\right\|_2\\
    =~& \left\|a(Z,\hat g) - a(Z,\hat g^{(-l_2)})\right\|_2
\end{align}
where we also used the fact that $\tilde{Z}_{l_1}$ only appears in the training sets of $\hat{g}^{(-l_1)}$ and $\hat{g}^{(-l_1, -l_2)}$ and we can simply rename it to $Z_{l_1}$, as they are identically distributed and both independent from all the other data points. Invoking the second of our stability conditions for $a$, we have that: $ 
    \max_{l_1\neq l_2}\left\|K_{l_1}-K_{l_1}^{(l_2)}\right\|_2 = o(n^{-1/2})$.
Hence, $J_{3,n}=o(1)$, which completes the proof.
\end{proof}

\subsection{Tightness of Stability Condition}\label{bla}

We present here an example that shows that, without further structural constraints (on the moment or the function space $\mcG$), the stability condition we impose is required for stochastic equicontinuity condition to hold. Let $(X_i)\overset{i.i.d}{\sim}\rm{unif}[0,1]$ be i.i.d uniform random variables. We set $Y_i:=\mathbb{I}(X_i\le 0.5)$ and set $Z_i:=(X_i,Y_i)$. For any $x\in [0,1]$, we define $c(Z_{1:n},x):=\argmin_{i\le n} |X_i-x|$ the function that returns the index of the nearest example to $x$ in $\{X_1,\dots,X_n\}$ and note that the quantity $Y_{c(Z_{1:n},x)}$ is its nearest neighbour estimator. Let $$\hat{g}(Z_{1:n})(x,y):=n^{1/6}\mathbb{I}(y\ne Y_{c(Z_{1:n},x)}),\qquad \nu(Z,g):= g(Z)^3.$$
We remark that $\nu(\cdot,\hat g)$ does not satisfy our stability conditions as  $ \|\nu(Z,\hat g)-\nu(Z,\hat g^{(-1)})\|_2$ is exactly of order $1/\sqrt{n}$, neither does it respect the stochastic equicontinuity property. 

\begin{lemma}\label{lem:countereg}
Let $(X_i)\overset{i.i.d}{\sim}\rm{unif}[0,1]$ be i.i.d uniform random variables. We set $Y_i:=\mathbb{I}(X_i\le 0.5)$ and set $Z_i:=(X_i,Y_i)$. There are constants $C,c>0$ such that 
\begin{itemize}
\item Set $g_0(Z):= 0$ then we have $\|\hat g(Z)-g_0(Z)\|^2_2\rightarrow 0$
    \item  $\frac{C}{\sqrt{n}}\ge \|\nu(Z,\hat g)-\nu(Z,\hat g^{(-1)})\|_2\ge \frac{c}{\sqrt{n}}$
    \item $\sqrt{n}|V(\hat g)-V( g_0)-\big(V_n(\hat g)-V_n(g_0)\big)|\not\xrightarrow{P}0. $
\end{itemize}
\end{lemma}

\section{Application: Bagging Estimators}\label{sec:baggapp}
We remark that if the functions $a$ and $\nu$ satisfy certain $L^p$-Lipchitz conditions then the stability conditions of \cref{lem:stoc-eq-jac} are implied by the algorithmic stability of the estimator $\hat g$. Those conditions are only marginally stronger than the condition of mean-squared-continuity found in \cref{lem:stoc-eq-jac} \begin{corollary}
\label{simpler}
Fix any constant $r>1.$ Suppose that there is $L<\infty$ such that the estimation algorithm satisfies the following uniform $L^{2r}$-continuity condition: for all $i,j\in [p]$ and $l\in [n]$
\begin{align}&\label{morgane=bed}
    \E[(a_{i,j}(Z_l;\hat g) - a_{i,j}(Z_l;\hat{g}^{(-l)}))^{2}]\leq~ L\cdot \mathbb{E}\Big[\sup_x\|\hat g(x)-\hat g^{(-l)}(x)\|_2^{2r}\Big]^{1/r} \\& \E[(\nu_{i}(Z_l;\hat g) - \nu_{i}(Z_l;\hat{g}^{(-l)}))^{2}]\leq~ L\cdot \mathbb{E}\Big[\sup_x\|\hat g(x)-\hat g^{(-l)}(x)\|_2^{2r}\Big]^{1/r}\qquad \tag{$L^{2r}$-Continuity}
    \\&
    \E[(a_{i,j}(Z;\hat g) - a_{i,j}(Z;\hat{g}^{(-l)}))^{2}]\leq~ L\cdot \mathbb{E}\Big[\sup_x\|\hat g(x)-\hat g^{(-l)}(x)\|_2^{2r}\Big]^{1/r} \\& \E[(\nu_{i}(Z;\hat g) - \nu_{i}(Z;\hat{g}^{(-l)}))^{2}]\leq~ L\cdot \mathbb{E}\Big[\sup_x\|\hat g(x)-\hat g^{(-l)}(x)\|_2^{2r}\Big]^{1/r}.
\end{align}Suppose in addition that 
\begin{align}\label{algo_stab}
\max_{l\le n}\mathbb{E}_{Z_{1:n}}\Big[\sup_x\|\hat g(x)-\hat g^{(-l)}(x)\|_2^{2r}\Big]^{1/2r}=o\Big(n^{-1/2}\Big) \tag{Algorithmic Stability}
\end{align}
and if in addition mean-squared-continuity condition is satisfied, 
 then the conditions of \cref{lem:stoc-eq-jac} are satisfied.

\end{corollary} The uniform $L^{2r}$-continuity conditions are going to be satisfied by most moment functions $m(\cdot;\cdot,\cdot)$. We notably show in the appendix that \cref{sleep2}, \cref{sleep3} and \cref{sleep1} with general moment conditions satisfy our conditions. 

The condition of \cref{algo_stab} is a commonly assumed condition in recent literature~\cite{austern2020asymptotics, bayle2020cross} and is satisfied by various  regularized empirical risk minimization estimators and stochastic gradient descent estimators~\cite{bousquet2002stability, hardt2016train}. Notably, it is satisfied by bagged estimators.
We show in \cref{bag:sleepy} that under very general conditions a bagged ensemble of any machine learning estimator is stable.
Let $Z_1,\ldots,Z_n\in\mathcal{Z}$ be an independent and identically distributed (i.i.d.) sample of size $n$. We sample uniformly randomly without replacement from these observations repeatedly and independently, each time taking a sample of size $m$, for a number of $B$ times. We denote the resulting samples as
$$Z_{1:m}^b:=\left\{Z_1^b,\ldots,Z_m^b\right\},b\in[B].$$

Let $\hat{h}_m:\mathcal{Z}^m\rightarrow \mathbb{R}^P$ be a base machine learning estimator trained on $m$ observations. This base estimator can be tree, a CNN, a nearest-neighbour classifier, or any other type of machine learning estimator. The corresponding bagged estimator is
$$\hat{g}(\cdot)=\frac{1}{B}\sum_{b=1}^B\hat{h}_m(Z_{1:m}^b)(\cdot).$$

We will show that subject to $B$ and $m$ being sufficiently big and mild moment conditions on the base estimator $\hat{h}_m$, the bagged estimator satisfies condition \cref{algo_stab}.

\begin{theorem}\label{bag:sleepy}
Fix any constants $s,k\geq 2r$ such that $\frac{1}{s}+\frac{1}{k}=\frac{1}{2r}.$ Assume $B,m$ satisfy
$$m=o(\sqrt{n})\hspace{1cm}B>>m^{2/{k}}\cdot n^{1-\frac{2}{{k}}},$$
and assume the base estimator $\hat h$ has bounded moments:
$$\left\Vert
    \sup_x\left\|\hat{h}(Z_{1:m}^1)(x)\right\|_2\right\Vert_s\leq C$$
for some constant $C>0.$ Then \cref{algo_stab} is achieved:
$$\max_{l\le n}\left\|\sup_x\|\hat g(x)-\hat g^{(-l)}(x)\|_2\right\|_{2r}=o(n^{-1/2}).$$
Therefore if $a$ and $\nu $ satisfy the condition \cref{morgane=bed} then the condition (\ref{lem:stoc-eq-jac}) is satisfied.
\end{theorem}



In particular, we can take the base machine learning estimator to be the 1-nearest neighbor estimator. This specific choice leads to a bagged estimator that can be proved to satisfy both our stability conditions and our consistency rate condition in Theorem~\ref{thm:main}, in regression settings where covariates are of small intrinsic dimension.

Let $\left\{Z_i=\left(X_i, Y_i\right)\right\}_{i=1}^n$ be a sample of size $n$, drawn independently and identically distributed from $Z=(X,Y)$. Here $X\in\mathcal{X}\subset \mathbb{R}^D$ are known covariates and $Y\in\mathcal{Y}\subset \mathbb{R}^P$ is the response variable. As in our bagging setting, let $Z_{1:m}^1,\ldots,Z_{1:m}^B$ be B independent samples of size $m$ drawn without replacement from observations $\left\{Z_i\right\}_{i=1}^n$. A bagged 1-nearest neighbor (1-NN) estimator takes the following form:
$$\hat g(x):=\frac{1}{B}\sum_{b=1}^B\sum_{i=1}^n \mathbbm{1}_{\left\{X_i= S_b(x)\right\}}Y_i,$$
where $S_b(x)$ is the 1-NN of $x$ in the set $Z_{1:m}^b.$ The estimator is used to estimate the conditional expectation $g_0(x):=\mathbb{E}[Y\mid X=x].$

\begin{lemma}[Special Case of Theorem 3 of \cite{khosravi2019non}]\label{lem: knn}
Assume that
\begin{itemize}
    \item The marginal distribution $\mu$ of $X_1$ satisfies $(C,d)$-homogeneity on the ball $B(x,r)$:
$$\mu(B(x,r))\le C \alpha^{-d}\mu(B(x,\alpha r))\hspace{1cm}\forall \alpha\in(0,1)$$
for some $C,r>0$. Here $d$ is referred to as the intrinsic dimension of the distribution.
  \item The conditional expectation $\mathbb{E}[Y\mid X=x]$ is a Lipschitz function in the coordinates $x$.
  \item The response variable $Y$ is bounded in $L^\infty$-norm:
  $\left\|Y \right\|_\infty<\infty.$
\end{itemize} 
  Then the bagged 1-NN estimator $\hat g$ with $B\ge \frac{n}{m}$ satisfies the following convergence condition:
  $$\sqrt{\mathbb{E}\left[\left\|\hat{g}(X)-g_0(X)\right\|_2^2\right]}\le O\left(m^{-1/d}\right)+O\left(\sqrt{\frac{m P \log\log(Pn/m)}{n}}\right).$$
\end{lemma}

In particular, we note that when $0<d<2$ and $m=O(n^{\frac{1}{2}-\epsilon})$ for some $0<\epsilon\leq\frac{1}{2}-\frac{1}{4}d,$ we immediately have that $$\sqrt{\mathbb{E}\left[\left\|\hat{g}(X)-g_0(X)\right\|_2^2\right]}=o(n^{-1/4}),$$ achieving the convergence rate assumed in Theorem~\ref{thm:main}. Moreover, provided that we choose $B$ to be sufficiently large such that
$$B>>n^{\epsilon+\frac{1}{2}} \hspace{1cm} B\geq n^{1-\frac{2\epsilon+1}{k}},$$
the required stability conditions can also be satisfied.

\input{experiments.tex}

\bibliography{refs}
\bibliographystyle{plain}
\newpage
\appendix



\section{Proof of Theorem~\ref{thm:main}}

\begin{proof}
For any $g\in \mcG$, by the linearity of the moment with respect to $\theta$:
\begin{align}
    A(g) \left(\hat{\theta}-\theta_0\right) =~&  M(\hat{\theta}, g) - M(\theta_0, g)\\
    =~& M(\hat{\theta}, g) - M_n(\hat{\theta}, g) + M(\theta_0, g_0) - M(\theta_0, g) + M_n(\hat{\theta}, g).
\end{align}
Moreover, for any $g$, with $\|g-g_0\|_2=o_p(1)$:
\begin{align}
    A( g) \left(\hat{\theta}-\theta_0\right) =~&  A(g_0) \left(\hat{\theta}-\theta_0\right) + \left(A(g) - 
    A(g_0)\right)\left(\hat{\theta}-\theta_0\right)\\
    =~& A(g_0) \left(\hat{\theta}-\theta_0\right) + O\left(\|g-g_0\|_2\, \|\hat{\theta}-\theta_0\|_2\right)\\
    =~& A(g_0) \left(\hat{\theta}-\theta_0\right) + o_p\left(\|\hat{\theta}-\theta_0\|_2\right).
\end{align}
Thus for any $g$, with $\|g-g_0\|_2=o_p(1)$:
\begin{align}
    A(g_0) \left(\hat{\theta}-\theta_0\right) =~& M(\hat{\theta}, g) - M_n(\hat{\theta}, g) + M(\theta_0, g_0) - M(\theta_0, g) + M_n(\hat{\theta}, g) + o_p(\|\hat{\theta}-\theta_0\|_2).
\end{align}
Let $G_n(\theta, g) := M(\theta, g) - M_n(\theta, g)$, then we have:
\begin{align}
    A( g_0) \left(\hat{\theta}-\theta_0\right) =~& G_n(\hat{\theta}, g) + M(\theta_0, g_0) - M(\theta_0, g) + M_n(\hat{\theta}, g) + o_p(\|\hat{\theta}-\theta_0\|_2).
\end{align}
Applying the above for $g=\hat{g}$ and by the definition of $\hat{\theta}$, we have: 
\begin{align}
    A(g_0) \left(\hat{\theta}-\theta_0\right) =~& G_n(\hat{\theta}, \hat{g}) + M(\theta_0, g_0) - M(\theta_0, \hat{g}) + M_n(\hat{\theta}, \hat{g}) + o_p(\|\hat{\theta}-\theta_0\|_2)\\
    =~& G_n(\hat{\theta}, \hat{g}) + M(\theta_0, g_0) - M(\theta_0, \hat{g}) + o_p(n^{-1/2} + \|\hat{\theta}-\theta_0\|_2).
\end{align}
Applying Neyman orthogonality and bounded second derivative of the moment with respect to $g$:
\begin{align}
    M(\theta_0, g_0) - M(\theta_0, g) = D_g\, M(\theta_0, g_0) [g_0 - g] + O\left(\|g-g_0\|_2^2\right) = O\left(\|g-g_0\|_2^2\right) = o_p(n^{-1/2}).
\end{align}
Thus we have that:
\begin{align}
    A(g_0) \left(\hat{\theta}-\theta_0\right) =~& G_n(\hat{\theta}, \hat{g}) + o_p(n^{-1/2} + \|\hat{\theta}-\theta_0\|_2).
\end{align}
Now we decompose the empirical process part into an asymptotically normal component and asymptotically equicontinuous parts that converge to zero in probability:
\begin{align}
G_n(\hat{\theta}, \hat{g}) =~& G_n(\theta_0, g_0) + \left(G_n(\hat{\theta}, \hat{g}) - G_n(\theta_0, \hat{g})\right) + \left(G_n(\theta_0, \hat{g}) - G_n(\theta_0, g_0)\right).
\end{align}

By the linearity of the moment, the middle term can be written as:
\begin{align}
    G_n(\hat{\theta}, \hat{g}) - G_n(\theta_0, \hat{g}) =~& \left(A(\hat{g}) - A_n(\hat{g})\right)'(\hat{\theta}-\theta_0).
\end{align}
Note that by a triangle inequality:
\begin{align}
   \left\|A(\hat{g})-A_n(\hat{g})\right\|_{op} \leq~& \left\|A(g_0)-A_n(g_0)\right\|_{op} + \left\|A(\hat{g}) - A(g_0) - (A_n(\hat{g}) - A_n(g_0))\right\|_{op}
\end{align}
Note that the first quantity is a simple centered empirical process and hence assuming that $a_{i,j}(Z;g_0)$ has bounded variance, by classical results in empirical process theory we have that:
\begin{align}
    \left\|A(g_0)-A_n(g_0)\right\|_{op} = o_p(1)
\end{align}
Moreover, by our stochastic equicontinuity condition we have that:
\begin{align}
    \left\|A(\hat{g}) - A(g_0) - (A_n(\hat{g}) - A_n(g_0))\right\|_{op} = o_p(n^{-1/2}) = o_p(1).
\end{align}
Thus we get that $\|A(\hat{g}) - A_n(\hat{g})\|_{op}=o_p(1)$, and therefore:
\begin{align}
    G_n(\hat{\theta}, \hat{g}) - G_n(\theta_0, \hat{g}) =~& o_p\left(\|\hat{\theta}-\theta_0\|_2\right).
\end{align}

Moreover, since by our stochastic equicontinuity conditions:
\begin{align}
    \sqrt{n} \left\|A(\hat{g}) - A(g_0) - (A_n(\hat{g}) - A_n(g_0))\right\|_{op} = o_p(1)\\
    \sqrt{n} \left\|V(\hat{g}) - V(g_0) - (V_n(\hat{g}) - V_n(g_0))\right\|_2 = o_p(1)
\end{align}
we have by triangle inequality, the definition of the operator norm, and the fact that $\|\theta_0\|_2=O(1)$ that:
\begin{align}
    \left\|G_n(\theta_0, \hat{g}) - G_n(\theta_0, g_0)\right\|_2\le~&\left\|A(\hat{g}) - A(g_0) - (A_n(\hat{g}) - A_n(g_0))\right\|_{op}\, \|\theta_0\|_2 \\
    ~& + \left\|V(\hat{g}) - V(g_0) - (V_n(\hat{g}) - V_n(g_0))\right\|_2\\
    =~& o_p(n^{-1/2}). 
\end{align}

Thus we can conclude that:
\begin{align}
    A(g_0) \left(\hat{\theta}-\theta_0\right) =~& G_n(\theta_0, g_0) + o_p(n^{-1/2} + \|\hat{\theta}-\theta_0\|_2).
\end{align}
Assuming that the inverse $A(g_0)^{-1}$ exists, we can re-arrange to:
\begin{equation}
    \hat{\theta}-\theta_0 = A(g_0)^{-1} G_n(\theta_0, g_0) + o_p\left(n^{-1/2} + \|\hat{\theta}-\theta_0\|_2\right).
\end{equation}

Since $G_n(\theta_0,g_0)$ is a mean-zero empirical process, we have that $\|G_n(\theta_0,g_0)\|_2 = O_p(n^{-1/2})$. Thus the above equation implies that $\|\hat{\theta}-\theta_0\|_2 = O_p(n^{-1/2})$. Thus we get:
\begin{equation}
    \hat{\theta}-\theta_0 = A(g_0)^{-1} G_n(\theta_0, g_0) + o_p\left(n^{-1/2}\right)
\end{equation}
or equivalently that:
\begin{equation}
    \sqrt{n} \left(\hat{\theta}-\theta_0\right) = \sqrt{n} A(g_0)^{-1} G_n(\theta_0, g_0) + o_p\left(1\right).
\end{equation}
The first term converges in distribution to the claimed normal limit by invoking the Central Limit Theorem. Thus the theorem follows by Slutsky's theorem.
\end{proof}

\section{Proof of Lemma~\ref{lem:countereg}}

\begin{proof}
Before diving into the proof, recall that $c(Z_{1:n},x) = \argmin_{i \leq n} \abs{X_i - x}$, and define:
\begin{align}
    c^*_1 &:= c(Z_{1:n},\frac{1}{2}), \\
    c^*_2 &:= \begin{cases}
        \displaystyle\argmin_{i\le n \text{ s.t } X_i\le \frac{1}{2}} \abs{\frac{1}{2}-X_i} &\text{if }X_{c^*_1}> \frac{1}{2},\\
        \displaystyle\argmin_{i\le n \text{ s.t } X_i> \frac{1}{2}} \abs{\frac{1}{2}-X_i} &\text{if } X_{c^*_1}\le \frac{1}{2}.
    \end{cases}.
\end{align}
That is, we let $c_1^*$ be the index of the nearest example in $\{X_1,\ldots, X_n\}$ to $\frac{1}{2}$, and let $c_2^*$ be the index of the nearest example to $\frac{1}{2}$ on the other side of $\frac{1}{2}$ from.

A new observation $Z=(X, Y)$, where $X\sim\mathrm{unif}[0,1]$ and $Y=\mathbb{I}(X\le 0.5)$,  will be misclassified if $Y_{c(Z_{1:n},X)}$ is different from $Y$.
Therefore it is mislabeled if it falls in the following set:
\begin{equation}
 \mathcal{E}\Bigl(X_{c_1^*},X_{c_2^*}\Bigr) := \begin{cases}
    \bigl[\frac{1}{2}, \frac{1}{2}\bigl(X_{c_1^*}+X_{c_2^*}\bigr)\bigr] & \text{if } X_{c_1^*}\le \frac{1}{2},\\
    \bigl[\frac{1}{2}\bigl(X_{c_1^*}+X_{c_2^*}\bigr),\frac{1}{2} \bigr] & \text{otherwise.}
    \end{cases}
\end{equation}
For a given pair of random variables $X_{i_1},X_{i_2}$, we write:
\begin{equation}
    \lambda_{X_{i_1,i_2}}:= n\mathbb{P}\big(X\in \mathcal{E}(X_{i_1},X_{i_2}) \mid X_{i_1}, X_{i_2}\big).
\end{equation}

We note $B_1:=\{i\le n|X_i\le 1/2\}$. 

Now we remark that if $X_{c_1^*}\le 0.5$ then $\Big|1-\Big[X_{c_1^*}+X_{c_2^*}\Big]\Big|=X_{c_1^*}+X_{c_2^*}-1$ and $$X_{c_2^*}-(1-X_{c_1^*})|X_{c_1^*}\sim \min_{i\in[n]\setminus B_1 } U_i(X_{c_1^*})$$ where $U_i(X_{c_1^*})|X_{c_1^*}\sim_{i.i.d}\rm{unif}[0,X_{c_1^*}].$ Here $\sim$ means "has the same distribution as."

Therefore we have
\begin{align}
    \mathbb{P}\Big(X_{c_1^*}+X_{c_2^*}-1\le \frac{2t}{n}\Big|X_{c_1^*}\le 0.5, X_{c_1^*}\Big)
    &=\mathbb{P}\Big(\min_{i\in[n]\setminus B_1 } U_i(X_{c_1^*})\le \frac{2t}{n}\Big|X_{c_1^*}\le 0.5,X_{c_1^*}\Big)
    \\&=1-\Big(1-\frac{2t}{nX_{c_1^*}}\Big)^{n-|B_1|},
\end{align}
where $|B_1|$ denotes the cardinality of set $B_1.$

Therefore as $X_{c_1^*}\rightarrow 0.5$ and $|B_1|\rightarrow n/2$ we have
\begin{align}
    P\Big(X_{c_1^*}+X_{c_2^*}-1\le \frac{2t}{n}\Big|X_{c_1^*}\le 0.5, X_{c_1^*}\Big)
    \rightarrow 1-e^{-2t}.
\end{align}Similarly we remark that if $X_{c_1^*}>0.5$ then $\Big|1-\Big[X_{c_1^*}+X_{c_2^*}\Big]\Big|=1-X_{c_1^*}-X_{c_2^*}$ and $$1-X_{c_1^*}-X_{c_2^*}|X_{c_1^*}\sim \min_{i\in B_1 } U_i^{bis}(X_{c_1^*})$$ where $U_i^{bis}(X_{c_1^*})|X_{c_1^*}\sim_{i.i.d}\rm{unif}[0,1-X_{c_1^*}]$.

Therefore we have
\begin{align}
    \mathbb{P}\Big(1-X_{c_1^*}-X_{c_2^*}\le \frac{2t}{n}\Big|X_{c_1^*}>0.5,X_{c_1^*}\Big)&=\mathbb{P}\Big(\min_{i\in B_1 } U_i^{bis}(X_{c_1^*})\le \frac{2t}{n}\Big|X_{c_1^*}>0.5,X_{c_1^*}\Big)
    \\&=1-\Big(1-\frac{2t}{n(1-X_{c_1^*})}\Big)^{|B_1|}.
\end{align}

Therefore as $X_{c_1^*}\rightarrow 0.5$ and $|B_1|\rightarrow n/2$ we have
\begin{align}
P\Big(1-X_{c_1^*}-X_{c_2^*}\le \frac{2t}{n}\Big|X_{c_1^*}>0.5, X_{c_1^*}\Big)
    \rightarrow 1-e^{-2t}.
\end{align}

This directly implies that \begin{equation}
    \begin{split}\label{blan}
        \lambda_{X_{c^*_1,c^*_2}}=\frac{n}{2}\Big|1-\Big[X_{c_1^*}+X_{c_2^*}\Big]\Big| \xrightarrow{d} \mathrm{Exp}(2),
    \end{split}
 \end{equation}
where $\mathrm{Exp}(2)$ denotes an exponential distribution with rate parameter $2.$


Moreover, we now also show that the expectation $\mathbb{E}\left[\lambda_{X_{c^*_1,c^*_2}}\right]=O(1).$

As a first step, we note that we can write
\begin{align}
    &\mathbb{E}\left[\lambda_{X_{c^*_1,c^*_2}}\right]=n\mathbb{P}\left(X\in\mathcal{E}\Bigl(X_{c_1^*},X_{c_2^*}\Bigr)\right)\\
    &\le n\mathbb{P}\left(X\in\mathcal{E}\Bigl(X_{c_1^*},X_{c_2^*}\Bigr),\left| \frac{|B_1|}{n}-\frac{1}{2} \right|\le \frac{1}{4}\right)+n\mathbb{P}\left(\left| \frac{|B_1|}{n}-\frac{1}{2} \right|> \frac{1}{4}\right).
\end{align}

By Azuma's concentration inequality, we know that
$$n\mathbb{P}\left(\left| \frac{|B_1|}{n}-\frac{1}{2} \right|> \frac{1}{4}\right)\le 2ne^{-\frac{32}{n}}\to 0 \mbox{ as $n\to\infty$}.$$

To treat the other term in the sum, we have that
\begin{align}
    &n\mathbb{P}\left(X\in\mathcal{E}\Bigl(X_{c_1^*},X_{c_2^*}\Bigr),\left| \frac{|B_1|}{n}-\frac{1}{2} \right|\le \frac{1}{4}\right)\\
    &=n\mathbb{E}\left[\mathbb{P}\left(X\in\mathcal{E}\Bigl(X_{c_1^*},X_{c_2^*}\Bigr),\left| \frac{|B_1|}{n}-\frac{1}{2} \right|\le \frac{1}{4}\Bigg| X_{1:n}\right)\right]\\
    &=n\mathbb{E}\left[\mathbb{P}\left(X\in\mathcal{E}\Bigl(X_{c_1^*},X_{c_2^*}\Bigr)\Bigg| X_{1:n}\right)\mathbb{I}\left(\left| \frac{|B_1|}{n}-\frac{1}{2} \right|\le \frac{1}{4}\right)\right]\\
    &=n\mathbb{E}\left[\left|\mathcal{E}\Bigl(X_{c_1^*},X_{c_2^*}\Bigr)\right|\mathbb{I}\left(\left| \frac{|B_1|}{n}-\frac{1}{2} \right|\le \frac{1}{4}\right)\right]
\end{align}
where $X_{1:n}:=(X_1,\ldots,X_n)$ and  $\left|\mathcal{E}\Bigl(X_{c_1^*},X_{c_2^*}\Bigr)\right|$ denotes the length of the interval $\mathcal{E}\Bigl(X_{c_1^*},X_{c_2^*}\Bigr).$

Now by triangle inequality
\begin{align}
    &n\mathbb{E}\left[\left|\mathcal{E}\Bigl(X_{c_1^*},X_{c_2^*}\Bigr)\right|\mathbb{I}\left(\left| \frac{|B_1|}{n}-\frac{1}{2} \right|\le \frac{1}{4}\right)\right]\\
    &=n\mathbb{E}\left[\frac{1}{2}\left|1-\left(X_{c_1^*}+X_{c_2^*}\right)\right|\mathbb{I}\left(\left| \frac{|B_1|}{n}-\frac{1}{2} \right|\le \frac{1}{4}\right)\right]\\&\le \frac{n}{2}\mathbb{E}\left[\left|\frac{1}{2}-X_{c_1^*}\right|\mathbb{I}\left(\left| \frac{|B_1|}{n}-\frac{1}{2} \right|\le \frac{1}{4}\right)\right]+\frac{n}{2}\mathbb{E}\left[\left|\frac{1}{2}-X_{c_2^*}\right|\mathbb{I}\left(\left| \frac{|B_1|}{n}-\frac{1}{2} \right|\le \frac{1}{4}\right)\right]\\
    &=
    \frac{n}{2}\mathbb{E}\left[\min_{i\in B_1}U_i\cdot\mathbb{I}\left(\left| \frac{|B_1|}{n}-\frac{1}{2} \right|\le \frac{1}{4}\right)\right]+\frac{n}{2}\mathbb{E}\left[\min_{i\in [n]\setminus B_1}U_i\cdot\mathbb{I}\left(\left| \frac{|B_1|}{n}-\frac{1}{2} \right|\le \frac{1}{4}\right)\right]\\
    &=n\mathbb{E}\left[\min_{i\in B_1}U_i\cdot\mathbb{I}\left(\left| \frac{|B_1|}{n}-\frac{1}{2} \right|\le \frac{1}{4}\right)\right]\\
    &\le n\mathbb{E}\left[\min_{i\in B_1}U_i\cdot\mathbb{I}\left(\left| B_1 \right|\ge \frac{n}{4}\right)\right].
\end{align}
where $U_i:=\left|\frac{1}{2}-X_i\right|\sim_{\rm{i.i.d.}}\mathrm{unif}[0,0.5],$ and the penultimate line follows from symmetry.

We can express the expectation in terms of integrals of tail probabilities as
\begin{align}
    &n\mathbb{E}\left[\min_{i\in B_1}U_i\cdot\mathbb{I}\left(\left| B_1 \right|\ge \frac{n}{4}\right)\right]=n\mathbb{E}\left[\mathbb{E}\left[\min_{i\in B_1}U_i\cdot\mathbb{I}\left(\left| B_1 \right|\ge \frac{n}{4}\right)\ \Big|\ |B_1|\right]\right]\\
    &=n\mathbb{E}\left[\int_{0}^\infty\mathbb{P}\left(\min_{i\in B_1}U_i\cdot\mathbb{I}\left(\left| B_1 \right|\ge \frac{n}{4}\right)\ge t\ \Big|\ |B_1|\right)dt\right]\\
    &=n\mathbb{E}\left[\frac{1}{|B_1|}\int_{0}^\infty\mathbb{P}\left(\min_{i\in B_1}U_i\cdot\mathbb{I}\left(\left| B_1 \right|\ge \frac{n}{4}\right)\ge \frac{t}{|B_1|}\ \Big|\ |B_1|\right)dt\right].
\end{align}

Here 
\begin{align}
    &\mathbb{P}\left(\min_{i\in B_1}U_i\cdot\mathbb{I}\left(\left| B_1 \right|\ge \frac{n}{4}\right)\ge \frac{t}{|B_1|}\ \Big|\ |B_1|\right)\\
    &=\left(1-\frac{2t}{|B_1|}\right)^{|B_1|}\cdot\mathbb{I}\left(\left| B_1 \right|\ge \frac{n}{4}\right)\le e^{-2t}\cdot\mathbb{I}\left(\left| B_1 \right|\ge \frac{n}{4}\right)
\end{align}
where we have used the inequality that $1-x\le e^{-x}$ for all $x$.

Hence, we have
\begin{align}
    &n\mathbb{E}\left[\frac{1}{|B_1|}\int_{0}^\infty\mathbb{P}\left(\min_{i\in B_1}U_i\cdot\mathbb{I}\left(\left| B_1 \right|\ge \frac{n}{4}\right)\ge \frac{t}{|B_1|}\ \Big|\ |B_1|\right)dt\right]\\
    &\le n\mathbb{E}\left[\frac{1}{|B_1|}\int_{0}^\infty e^{-2t}\cdot\mathbb{I}\left(\left| B_1 \right|\ge \frac{n}{4}\right)dt\right]=\frac{n}{2}\mathbb{E}\left[\frac{1}{|B_1|}\cdot\mathbb{I}\left(\left| B_1 \right|\ge \frac{n}{4}\right)\right]\\
    &\le \frac{n}{2}\mathbb{E}\left[\frac{4}{n}\cdot\mathbb{I}\left(\left| B_1 \right|\ge \frac{n}{4}\right)\right]=2\mathbb{P}\left(\left| B_1 \right|\ge \frac{n}{4}\right)\le 2.
\end{align}

Altogether, we have shown that
\begin{equation}\label{lambda}
    \mathbb{E}\left[\lambda_{X_{c^*_1,c^*_2}}\right]=O(1).
\end{equation}

The key point is to note that
\begin{align}
    \|\hat g(Z_{1:n})(Z)\|_2^2&=n^{1/3}\mathbb{P}(Y\ne Y_{c(Z_{1:n},X)})
  \\&\le  n^{1/3}\mathbb{P}(X\in\mathcal{E}(X_{c_1^*},X_{c_2^*}))\xrightarrow{(a)} 0.
\end{align}
where (a) comes from realizing that 
$$\mathbb{P}(X\in\mathcal{E}(X_{c_1^*},X_{c_2^*}))=\frac{1}{n}\mathbb{E}\left[\lambda_{X_{c^*_1,c^*_2}}\right].$$
Therefore we proved the first point.
Moreover if we denote \begin{align}
    c^{*(-1)}_1 &:= c(Z^{(-1)}_{1:n},\frac{1}{2}), \\
    c^{*(-1)}_2 &:= \begin{cases}
        \displaystyle\argmin_{i\le n \text{ s.t } X^{(-1)}_i\le \frac{1}{2}} \abs{\frac{1}{2}-X^{(-1)}_i} &\text{if }X^{(-1)}_{c^{*(-1)}_1}> \frac{1}{2},\\
        \displaystyle\argmin_{i\le n \text{ s.t } X^{(-1)}_i> \frac{1}{2}} \abs{\frac{1}{2}-X^{(-1)}_i} &\text{if } X^{(-1)}_{c^{*(-1)}_1}\le \frac{1}{2}
    \end{cases},
\end{align} 
where $Z_{1:n}^{(-1)}$ is $Z_{1:n}$ with the first observation replaced with an independent copy $\tilde{Z}_1=(\tilde{X}_1,\tilde{Y}_1)$, then we remark that
\begin{align}
    &\mathbb{P}\big( c^{*(-1)}_1\ne c^*_1 ~\rm{or}~ c^{*(-1)}_2\ne c^*_2)\\
    &\le \mathbb{P}\big( c^{*(-1)}_1\ne c^*_1) +\mathbb{P}( c^{*(-1)}_2\ne c^*_2)\\
    &\le \mathbb{P}(1=c^*_1)+\mathbb{P}(1=c^{*(-1)}_1)+\mathbb{P}(1=c^*_2)+\mathbb{P}(1=c^{*(-1)}_2)\overset{(b)}{\le }\frac{4}{n}
\end{align}
where (b) comes from symmetry (each observation has an equal chance of being $c^*_1$, for example). Moreover we note that if neither $c^{*(-1)}_1\ne c^*_1 ~\rm{nor}~ c^{*(-1)}_2\ne c^*_2$ then we have $\mathcal{E}(X^{(-1)}_{c^{*(-1)}_1}, X^{(-1)}_{c^{*(-1)}_2})=\mathcal{E}(X_{c_1^*}, X_{c_2^*})$. 

Now for ease of notation denote
$$X_{c_1^*,c_2^*}:=\left(X_{c_1^*}, X_{c_2^*}\right),$$

$$X_{c^{*(-1)}_1,c^{*(-1)}_2}:=\left(X^{(-1)}_{c^{*(-1)}_1}, X^{(-1)}_{c^{*(-1)}_2}\right),$$

$$X_{c_1^*,c_2^*,c^{*(-1)}_1,c^{*(-1)}_2}:=\left(X_{c_1^*}, X_{c_2^*},X^{(-1)}_{c^{*(-1)}_1}, X^{(-1)}_{c^{*(-1)}_2}\right),$$
    and
    $$\mathcal{E}\left(X_{c_1^*,c_2^*,c^{*(-1)}_1,c^{*(-1)}_2}\right):=\mathcal{E}(X_{c_1^*}, X_{c_2^*}
    )\triangle \mathcal{E}\Big(X^{(-1)}_{c^{*(-1)}_1}, X^{(-1)}_{c^{*(-1)}_2}\Big).$$

Then we have that there is $C<\infty$ such that \begin{align}
    &\|\nu(Z,\hat g)-\nu(Z,\hat g^{(-1)})\|_2=\sqrt{n}\sqrt{\mathbb{P}\Big(X\in \mathcal{E}\left(X_{c_1^*,c_2^*,c^{*(-1)}_1,c^{*(-1)}_2}\right)\Big)}\\
    &=\sqrt{n\mathbb{E}\left[\mathbb{P}\left(X\in \mathcal{E}\left(X_{c_1^*,c_2^*,c^{*(-1)}_1,c^{*(-1)}_2}\right)\Big| X_{c_1^*,c_2^*,c^{*(-1)}_1,c^{*(-1)}_2}\right)\mathbb{I}\left(c^{*(-1)}_1\ne c^*_1 ~\rm{or}~ c^{*(-1)}_2\ne c^*_2\right)\right]}\\
    &\le\sqrt{n\mathbb{E}\left[\left(\left| \mathcal{E}(X_{c_1^*,c_2^*}
    )\right|
    +
    \left|\mathcal{E}\Big(X_{c^{*(-1)}_1,c^{*(-1)}_2}\Big)\right|
    \right)\mathbb{I}\left(c^{*(-1)}_1\ne c^*_1 ~\rm{or}~ c^{*(-1)}_2\ne c^*_2\right)\right]}
    \\
    &\le\sqrt{2n\mathbb{E}\left[\left| \mathcal{E}(X_{c_1^*,c_2^*}
    )\right|
    \mathbb{I}\left(c^{*(-1)}_1\ne c^*_1 ~\rm{or}~ c^{*(-1)}_2\ne c^*_2\right)\right]} \mbox{ by symmetry}
    \\
    &\le \sqrt{2n\mathbb{E}\left[\left| \mathcal{E}(X_{c_1^*,c_2^*}
    )\right|
    \left(
    \mathbb{I}\left(
    c^*_1=1\right)
    + \mathbb{I}\left(
    c^{*(-1)}_1=1 \right)
    + \mathbb{I}\left(
    c^*_2=1 \right)
    + \mathbb{I}\left(
    c^{*(-1)}_2=1
    \right)
    \right)
    \right]}\\
    &\overset{(c)}{=} \sqrt{2n
    \mathbb{P}\big(X\in \mathcal{E}(X_{c_1^*,c_2^*})\big)
    \left(
    \mathbb{P}\left(
    c^*_1=1\right)
    + \mathbb{P}\left(
    c^{*(-1)}_1=1 \right)
    + \mathbb{P}\left(
    c^*_2=1 \right)
    + \mathbb{P}\left(
    c^{*(-1)}_2=1
    \right)
    \right)
    }
    \\&= \frac{ \sqrt{8}}{\sqrt{n}} \sqrt{\mathbb{E}\left[\lambda_{X_{c_1^*,c_2^*}}\right]}\overset{(d)}{\le}\frac{C}{\sqrt{n}}
\end{align}
where to get (c) we exploited independence of $(X_{c_1^*}, X_{c_2^*})$ and the events $\{c^*_1=1\}$, $\{c^*_2=1\}$, $\{c^{*(-1)}_1=1\}$, $\{c^{*(-1)}_2=1\}$ and where to get (d) we exploited \cref{lambda}.


Moreover we also notice that 
\begin{align}
    & \mathbb{P}\left( c^{*(-1)}_1\ne c^*_1\right)\\
    &\ge \mathbb{P}\left( c^*_1=1,\tilde{X}_1\not\in[
    \frac{1}{2}-\left|\frac{1}{2}-X_{c_3^*}\right|,
    \frac{1}{2}+\left|\frac{1}{2}-X_{c_3^*}\right|
    ]\right)\\
    &=\mathbb{P}\left( c^*_1=1\right)\mathbb{P}\left(
    \tilde{X}_1\not\in[
    \frac{1}{2}-\left|\frac{1}{2}-X_{c_3^*}\right|,
    \frac{1}{2}+\left|\frac{1}{2}-X_{c_3^*}\right|
    ]\right) \mbox{ by independence}\\
    &=\frac{1}{n}\mathbb{E}\left[1-2\left|\frac{1}{2}-X_{c_3^*}\right|\right]=\frac{1}{n}-\frac{2}{n}\mathbb{E}\left[\left|\frac{1}{2}-X_{c_3^*}\right|\right]
\end{align}
where $c_3^*:= \arg\min_{i\in[n]\setminus\{c_1^*\}}\left|
X_i-\frac{1}{2}
\right|$ is the index of the second nearest neighbor of $\frac{1}{2}$ among $X_{1:n}.$ Note that $c_3^*$ is not necessarily equal to $c_2^*$, since the definition of $c_2^*$ requires $X_{c_2^*}$ to be on the other side of $\frac{1}{2}$ from $X_{c_1^*}$ while that of $c_3^*$ does not.

By our knowledge of the expectation of the second order statistic among i.i.d. uniform random variables, we obtain that
$$\frac{1}{n}-\frac{2}{n}\mathbb{E}\left[\left|\frac{1}{2}-X_{c_3^*}\right|\right]=\frac{1}{n}-\frac{2}{n}\cdot \frac{1}{2}\cdot\frac{2}{n+1}=\frac{1}{n}-\frac{2}{n(n+1)}.$$



Therefore, similarly to before, we also have that there is a constant $\tilde{c},c>0$ such that
\begin{align}
    &\|\nu(Z,\hat g)-\nu(Z,\hat g^{(-1})\|_2=\sqrt{n}\sqrt{\mathbb{P}\Big(X\in \mathcal{E}(X_{c_1^*}, X_{c_2^*})\triangle \mathcal{E}\Big(X^{(-1)}_{c^{*(-1)}_1}, X^{(-1)}_{c^{*(-1)}_2}\Big)\Big)}\\
    &\ge \sqrt{n}\sqrt{\mathbb{P}\Big(X\in \mathcal{E}(X_{c_1^*}, X_{c_2^*})\triangle \mathcal{E}\Big(X^{(-1)}_{c^{*(-1)}_1}, X^{(-1)}_{c^{*(-1)}_2}\Big)\Big|c^{*(-1)}_1\ne c^*_1\Big)\mathbb{P}\Big(c^{*(-1)}_1\ne c^*_1\Big)}\\
    &\ge \sqrt{n}\sqrt{\mathbb{P}\Big(X\in \mathcal{E}(X_{c_1^*}, X_{c_2^*})\triangle \mathcal{E}\Big(X^{(-1)}_{c^{*(-1)}_1}, X^{(-1)}_{c^{*(-1)}_2}\Big)\Big|c^{*(-1)}_1\ne c^*_1, c^{*(-1)}_2= c^*_2 \Big)\mathbb{P}\Big(c^{*(-1)}_1\ne c^*_1\Big)}\\
    &=\sqrt{n}\sqrt{\mathbb{E}\Big[
    \frac{1}{2}\left|X_{c_1^*}-X^{(-1)}_{c^{*(-1)}_1}\right|
    \Big|c^{*(-1)}_1\ne c^*_1, c^{*(-1)}_2= c^*_2 \Big]\mathbb{P}\Big(c^{*(-1)}_1\ne c^*_1\Big)}\\
    &\ge \tilde{c} \sqrt{\mathbb{P}\Big(c^{*(-1)}_1\ne c^*_1\Big)}\ge\frac{c}{\sqrt{n}}.
\end{align}
Therefore we proved the second point. The third point follows because
\begin{align}
\sqrt{n}\Big[V_n(\hat g)-V_n(g_0)-\big(V(\hat g)-V(g_0)\big)\Big] &{=}-\sqrt{n}V(\hat g){=}-\sqrt{n}V(\hat g)= \lambda_{X_{c_1^*,c_2^*}}
\xrightarrow{d} \rm{Exp}(0.5)
\end{align}
where $V_n(g_0)=V(g_0)=0$ by definition, and $V_n(\hat g)$ since the nearest neighbor estimator evaluated at a training data point never misclassifies the point.
\end{proof}

\section{Proof of Corollary~\ref{simpler}}
\begin{proof}
We note that by monotonicity of $L^p$ norms, plugging the bound in \cref{algo_stab} into the right hand side terms of \cref{morgane=bed} gives the stability conditions in \cref{lem:stoc-eq-jac}. Corollary~\ref{simpler} then immediately follows.
\end{proof}

\section{Proof of Theorem~\ref{bag:sleepy}}
\begin{proof}

Denote ${Z}_{1:m,(-l)}^b, b\in\{1,\ldots,B\}$ as the corresponding bagged samples when the $l$-th data point $Z_l$ is replaced with an independent copy $\tilde{Z}_l.$
We have that for $l\in[n]:$
\begin{align}
    &\left\|\sup_x\|\hat g(x)-\hat g^{(-l)}(x)\|_2\right\|_{2r}\\
    &= \left\|\sup_x\left\|\frac{1}{B}\sum_{b=1}^B\left(\hat{h}(Z_{1:m}^b)(x)-\hat{h}({Z}_{1:m,(-l)}^b)(x)\right)\right\|_2\right\|_{2r}\\
    &= \left\|\sup_x\left\|\frac{1}{B}\sum_{b=1}^B\left(\hat{h}(Z_{1:m}^b)(x)-\hat{h}({Z}_{1:m,(-l)}^b)(x)\right)\mathbbm{1}_{\{\exists t\leq m \mbox{ s.t. }Z_t^b=Z_l\} }\right\|_2\right\|_{2r}.
\end{align}

\noindent The last equality was because
$$\hat{h}(Z_{1:m}^b)(x)-\hat{\theta}({Z}_{1:m,(-l)}^b)(x)\neq 0$$
only if
$$\exists t\leq m \mbox{ s.t. }Z_t^b=Z_l.$$

\noindent Fix any $l\in[n]$. To simplify notations, let
$$\nabla\hat{h}(Z^b)(x):=\hat{h}(Z_{1:m}^b)(x)-\hat{h}({Z}_{1:m,(-l)}^b)(x),$$
and let $A_b$ be the event $\{\exists t\leq m \mbox{ s.t. }Z_t^b=Z_l\}$. 

Then
$$\left\|\sup_x\|\hat g(x)-\hat g^{(-l)}(x)\|_2\right\|_{2r}=\left\|\sup_x\left\|\frac{1}{B}\sum_{b=1}^B\nabla\hat{h}(Z^b)(x)\mathbbm{1}_{A_b}\right\|_2\right\|_{2r}.$$


\noindent By triangle inequality and symmetry of distributions
\begin{align}
    &\left\|\sup_x\|\hat g(x)-\hat g^{(-l)}(x)\|_2\right\|_{2r}=\left\|\sup_x\left\|\frac{1}{B}\sum_{b=1}^B\nabla\hat{h}(Z^b)(x)\mathbbm{1}_{A_b}\right\|_2\right\|_{2r}\\
    &\le \frac{1}{B} \left\|\sum_{b=1}^B\sup_x\left\|\nabla\hat{h}(Z^b)(x)\mathbbm{1}_{A_b}\right\|_2\right\|_{2r}\\
    &\leq \frac{1}{B} \left\|\sum_{b=1}^B\left\{\sup_x\left\|\nabla\hat{h}(Z^b)(x)\mathbbm{1}_{A_b}\right\|_2-\mathbb{E}\left[\sup_x\left\|\nabla\hat{h}(Z^b)(x)\mathbbm{1}_{A_b}\right\|_2\Big|Z_1,\ldots,Z_n, \tilde{Z}_l\right]
    \right\}\right\|_{2r}\\
    &+\frac{1}{B}\sum_{b=1}^B\left\|\mathbb{E}\left[\sup_x\left\|\nabla\hat{h}(Z^b)(x)\mathbbm{1}_{A_b}\right\|_2\Big|Z_1,\ldots,Z_n, \tilde{Z}_l\right]\right\|_{2r}\\
    &\leq \frac{1}{B} \left\|\sum_{b=1}^B\left\{\sup_x\left\|\nabla\hat{h}(Z^b)(x)\mathbbm{1}_{A_b}\right\|_2-\mathbb{E}\left[\sup_x\left\|\nabla\hat{h}(Z^b)(x)\mathbbm{1}_{A_b}\right\|_2\Big|Z_1,\ldots,Z_n, \tilde{Z}_l\right]
    \right\}\right\|_{2r}\\
    &+\left\|\mathbb{E}\left[\sup_x\left\|\nabla\hat{h}(Z^1)(x)\mathbbm{1}_{A_1}\right\|_2\Big|Z_1,\ldots,Z_n, \tilde{Z}_l\right]\right\|_{2r}.
\end{align}

For ease of notations, denote
$$Z_{(l)}:=\left(Z_1,\ldots,Z_n, \tilde{Z}_l\right)$$
and denote
$$R_b:=\sup_x\left\|\nabla\hat{h}(Z^b)(x)\mathbbm{1}_{A_b}\right\|_2-\mathbb{E}\left[\sup_x\left\|\nabla\hat{h}(Z^b)(x)\mathbbm{1}_{A_b}\right\|_2\Big|Z_{(l)}\right].$$

\noindent For the first term, we have by tower law that
\begin{align}
    &\frac{1}{B} \left\|\sum_{b=1}^B\left\{\sup_x\left\|\nabla\hat{h}(Z^b)(x)\mathbbm{1}_{A_b}\right\|_2-\mathbb{E}\left[\sup_x\left\|\nabla\hat{h}(Z^b)(x)\mathbbm{1}_{A_b}\right\|_2\Big|Z_{(l)}\right]
    \right\}\right\|_{2r}
    =\frac{1}{B} \left\|\sum_{b=1}^BR_b\right\|_{2r}\\
    &=\frac{1}{B}\mathbb{E}\left[\left(\sum_{b=1}^BR_b\right)^{2r}\right]^{\frac{1}{2r}}=\frac{1}{B}\mathbb{E}\left[\mathbb{E}\left[\left(\sum_{b=1}^BR_b\right)^{2r}\Bigg|Z_{(l)}\right]\right]^{\frac{1}{2r}}.
\end{align}

\noindent To further simplify, we use the following lemma:
\begin{lemma}[Marcinkiewicz-Zygmund inequality~\cite{rio2009moment}]\label{lem:lpnorm}
Let $p\geq 1$. If $X_1,\ldots,X_n$ are i.i.d. random variables such that $\mathbb{E}[X_1]=0$, then there exists a constant $C_p$ such that
$$\Vert\frac{1}{\sqrt{n}}\sum_{i=1}^n X_i\Vert_p\leq C_p\Vert X_i\Vert_p.$$
\end{lemma}

\noindent Since $R_b,b\in[B]$ are i.i.d. conditional on $Z_{(l)}$, we then have that
\begin{align}
    &\frac{1}{B}\mathbb{E}\left[\mathbb{E}\left[\left(\sum_{b=1}^BR_b\right)^{2r}\Bigg|Z_{(l)}\right]\right]^{\frac{1}{2r}}\\
    & \leq \frac{1}{B}\cdot
    \mathbb{E}\left[\left(\sqrt{B}C_{2r}\right)^{2r}\mathbb{E}\left[R_b^{2r}\Bigg|Z_{(l)}\right]\right]^{1/2r}=\frac{C_{2r}}{\sqrt{B}}\cdot\left\Vert R_b\right\Vert_{2r} \mbox{ by tower law}.
\end{align}

\noindent Now also
\begin{align}
    &\left\|\mathbb{E}\left[\sup_x\left\|\nabla\hat{h}(Z^1)(x)\mathbbm{1}_{A_1}\right\|_2\Big|Z_1,\ldots,Z_n, \tilde{Z}_l\right]\right\|_{2r}\\
    &=\left\Vert\mathbb{P}\left(A_1\Big|Z_1,\ldots,Z_n, \tilde{Z}_l\right)
    \mathbb{E}\left[\sup_x\left\|\nabla\hat{h}(Z^1)(x)\right\|_2\Big|Z_1,\ldots,Z_n, \tilde{Z}_l,A_1\right]\right\Vert_{2r}\\
    &\hspace{3cm}\mbox{since $\nabla\hat{h}(Z^1)(x)=0$ for all $x$ on $A_1^c$}\\
    &=\left\Vert\mathbb{P}\left(A_1\right)
    \mathbb{E}\left[\sup_x\left\|\nabla\hat{h}(Z^1)(x)\right\|_2\Big|Z_1,\ldots,Z_n, \tilde{Z}_l,A_1\right]\right\Vert_{2r}\\
    &\hspace{3cm}\mbox{since $A_1$ is independent of $Z_1,\ldots,Z_n, \tilde{Z}_l$}\\
    &=\mathbb{P}\left(A_1\right)\left\Vert
    \mathbb{E}\left[\sup_x\left\|\nabla\hat{h}(Z^1)(x)\right\|_2\Big|Z_1,\ldots,Z_n, \tilde{Z}_l,A_1\right]\right\Vert_{2r}\\
   &\le \mathbb{P}\left(A_1\right)\left\Vert
    \sup_x\left\|\nabla\hat{h}(Z^1)(x)\right\|_2\right\Vert_{2r} \mbox{ by Jensen's inequality and tower law}\\
    &\leq \mathbb{P}\left(A_1\right)\cdot 2C \mbox{ by moment condition, since $2r\le s$.}
\end{align}

\noindent By union bound,
$$\mathbb{P}\left(A_1\right)\leq\sum_{t=1}^m \mathbb{P}\left(Z_t^1=Z_l\right)=\frac{m}{n}.$$

\noindent Therefore, altogether we obtain
\begin{align}
     &\left\|\sup_x\|\hat g(x)-\hat g^{(-l)}(x)\|_2\right\|_{2r}\\
    &\leq \frac{C_{2r}}{\sqrt{B}}\cdot\left\Vert R_b\right\Vert_{2r} + 2C\cdot\frac{m}{n}\\
    &\leq \frac{C_{2r}}{\sqrt{B}}\cdot\left\Vert\sup_x\left\|\nabla\hat{h}(Z^1)(x)\mathbbm{1}_{A_1}\right\|_2\right\Vert_{2r}+\frac{C_{2r}}{\sqrt{B}}\cdot\left\Vert\mathbb{E}\left[\sup_x\left\|\nabla\hat{h}(Z^b)(x)\mathbbm{1}_{A_b}\right\|_2\Big|Z_{(l)}\right]\right\Vert_{2r} + 2C\cdot\frac{m}{n}\\
    &\leq \frac{C_{2r}}{\sqrt{B}}\cdot\left\Vert\sup_x\left\|\nabla\hat{h}(Z^1)(x)\mathbbm{1}_{A_1}\right\|_2\right\Vert_{2r}+\frac{C_{2r}}{\sqrt{B}}\cdot2C\cdot\frac{m}{n}+2C\cdot\frac{m}{n}\\
    &\leq \frac{2C\cdot C_{2r}}{\sqrt{B}}\cdot\left\Vert\mathbbm{1}_{A_1 }\right\Vert_{k}+\frac{C_{2r}}{\sqrt{B}}\cdot2C\cdot\frac{m}{n}+2C\cdot\frac{m}{n}\\
    &= \frac{2C\cdot C_{2r}}{\sqrt{B}}\cdot\left(\mathbb{P}(A_1 )\right)^{1/k}+\frac{C_{2r}}{\sqrt{B}}\cdot2C\cdot\frac{m}{n}+2C\cdot\frac{m}{n}\\
    &\leq \frac{2C\cdot C_{2r}}{\sqrt{B}}\cdot\left(\frac{m}{n}\right)^{1/{k}}+\frac{C_{2r}}{\sqrt{B}}\cdot2C\cdot\frac{m}{n}+2C\cdot\frac{m}{n}.
\end{align}

Hence, we also have
$$\max_{l\le n}\left\|\sup_x\|\hat g(x)-\hat g^{(-l)}(x)\|_2\right\|_{2r}\leq \frac{2C\cdot C_{2r}}{\sqrt{B}}\cdot\left(\frac{m}{n}\right)^{1/{k}}+\frac{C_{2r}}{\sqrt{B}}\cdot2C\cdot\frac{m}{n}+2C\cdot\frac{m}{n}.$$

\noindent Assuming 
$$m=o(\sqrt{n})$$
and
$$B>>m^{2/{k}}\cdot n^{1-\frac{2}{{k}}},$$
this upper bound is of order $o(n^{-1/2})$:
$$\max_{l\le n}\left\|\sup_x\|\hat g(x)-\hat g^{(-l)}(x)\|_2\right\|_{2r}=o(n^{-1/2}).$$
\end{proof}

\begin{remark}
We could in fact relax the conditions in Theorem~\ref{bag:sleepy} by using Rosenthal's inequality instead of Marcinkiewicz-Zygmund inequality in the proof. Moreover, we can relax the bounded moments condition to restrict on $L^{2r}$ norm instead of on $L^{s}$ norm. This gives the following theorem.
\end{remark}

\begin{theorem}\label{thm:relaxbag}
Assume $B,m$ satisfy
$$m=o(\sqrt{n})\hspace{1cm}B>>m^{\frac{1}{2r-1}}\cdot n^{\frac{r-1}{2r-1}},$$
and assume the base estimator $\hat h$ has bounded moments:
$$\max_{l\le n}\left\Vert
    \sup_x\left\|\hat{h}(Z_{1:m}^1)(x)\right\|_2\right\Vert_{2r}\leq C$$
for some constant $C>0.$ Then \cref{algo_stab} is achieved:
$$\max_{l\le n}\left\|\sup_x\|\hat g(x)-\hat g^{(-l)}(x)\|_2\right\|_{2r}=o(n^{-1/2}).$$
Therefore if $a$ and $\nu $ satisfy the condition \cref{morgane=bed} then the condition (\ref{lem:stoc-eq-jac}) is satisfied.
\end{theorem}

\begin{proof}[Proof of Theorem~\ref{thm:relaxbag}]
We follow the proof of Theorem~\ref{bag:sleepy} and obtain
\begin{align}
    &\left\|\sup_x\|\hat g(x)-\hat g^{(-l)}(x)\|_2\right\|_{2r}=\left\|\sup_x\left\|\frac{1}{B}\sum_{b=1}^B\nabla\hat{h}(Z^b)(x)\mathbbm{1}_{A_b}\right\|_2\right\|_{2r}\\
    &\leq \frac{1}{B} \left\|\sum_{b=1}^BR_b\right\|_{2r}+2C\cdot \frac{m}{n}\\
    &=\frac{1}{B}\cdot\mathbb{E}\left[\mathbb{E}\left[\left(\sum_{b=1}^BR_b\right)^{2r}\Bigg|Z_{(l)}\right]\right]^{1/2r}+2C\cdot \frac{m}{n} \hspace{0.5cm}\mbox{ by tower law}.
\end{align}

We then use the following lemma.
\begin{lemma}[Rosenthal's inequality~\cite{ibragimov1998exact}]
Let $p\geq 1$. If $X_1,\ldots,X_n$ are i.i.d. random variables such that $\mathbb{E}[X_1]=0$, then there exists a constant $C_p$ such that
$$\Vert\frac{1}{\sqrt{n}}\sum_{i=1}^n X_i\Vert_p\leq C_p\left(\Vert X_i\Vert_2+n^{\frac{1}{p}-\frac{1}{2}}\Vert X_i\Vert_p\right).$$
\end{lemma}

Since $R_b,b\in[B]$ are i.i.d. conditional on $Z_{(l)}$, we then have that by symmetry of distributions
\begin{align}
    &\frac{1}{B}\cdot\mathbb{E}\left[\mathbb{E}\left[\left(\sum_{b=1}^BR_b\right)^{2r}\Bigg|Z_{(l)}\right]\right]^{1/2r}\\
    & \leq \frac{1}{B}\cdot
    \mathbb{E}\left[\left(\sqrt{B}C_{2r}\right)^{2r}\left(\left(\mathbb{E}\left[R_1^2\Big|Z_{(l)}\right]\right)^{1/2}+B^{\frac{1}{2r}-\frac{1}{2}}\left( \mathbb{E}\left[R_1^{2r}\Big|Z_{(l)}\right] \right)^{1/2r}\right)^{2r}\right]^{1/2r}
    \\
    &= \frac{C_{2r}}{\sqrt{B}}\cdot\left\Vert \left(\mathbb{E}\left[R_1^2\Big|Z_{(l)}\right]\right)^{1/2}+B^{\frac{1}{2r}-\frac{1}{2}}\left( \mathbb{E}\left[R_1^{2r}\Big|Z_{(l)}\right] \right)^{1/2r} \right\Vert_{2r}.
\end{align}

By triangle inequality, we have
\begin{align}
    &\frac{C_{2r}}{\sqrt{B}}\cdot\left\Vert \left(\mathbb{E}\left[R_1^2\Big|Z_{(l)}\right]\right)^{1/2}+B^{\frac{1}{2r}-\frac{1}{2}}\left( \mathbb{E}\left[R_1^{2r}\Big|Z_{(l)}\right] \right)^{1/2r} \right\Vert_{2r}\\
    &\leq \frac{C_{2r}}{\sqrt{B}}\cdot\left\Vert \left(\mathbb{E}\left[R_1^2\Big|Z_{(l)}\right]\right)^{1/2} \right\Vert_{2r}
    +
    \frac{C_{2r}}{\sqrt{B}}\cdot\left\Vert B^{\frac{1}{2r}-\frac{1}{2}}\left( \mathbb{E}\left[R_1^{2r}\Big|Z_{(l)}\right] \right)^{1/2r} \right\Vert_{2r}.
\end{align}

For further ease of notations let
$$R_{1,(1)}:=\sup_x\left\|\nabla\hat{h}(Z^1)(x)\mathbbm{1}_{A_1}\right\|_2,$$

$$R_{1,(2)}:=\mathbb{E}\left[\sup_x\left\|\nabla\hat{h}(Z^1)(x)\mathbbm{1}_{A_1}\right\|_2\Big|Z_1,\ldots,Z_n, \tilde{Z}_l\right].$$

Note that $$R_1=R_{1,(1)}-R_{1,(2)}.$$

Then also by triangle inequality and Jensen's inequality
\begin{align}
    &\frac{C_{2r}}{\sqrt{B}}\cdot\left\Vert \left(\mathbb{E}\left[R_1^2\Big|Z_{(l)}\right]\right)^{1/2} \right\Vert_{2r}\\
    &\le\frac{C_{2r}}{\sqrt{B}}\cdot\left\Vert \left(\mathbb{E}\left[R_{1,(1)}^2\Big|Z_{(l)}\right]\right)^{1/2} \right\Vert_{2r}+\frac{C_{2r}}{\sqrt{B}}\cdot\left\Vert \left(\mathbb{E}\left[R_{1,(2)}^2\Big|Z_{(l)}\right]\right)^{1/2} \right\Vert_{2r}\\
    &\leq \frac{2C_{2r}}{\sqrt{B}}\cdot\left\Vert \left(\mathbb{E}\left[R_{1,(1)}^2\Big|Z_{(l)}\right]\right)^{1/2} \right\Vert_{2r}\\
    &=\frac{2C_{2r}}{\sqrt{B}}\cdot\left\Vert \left(\mathbb{E}\left[\left(
    \sup_x\left\|\nabla\hat{h}(Z^1)(x)\mathbbm{1}_{A_1}\right\|_2
    \right)^2\Big|Z_{(l)}\right]\right)^{1/2} \right\Vert_{2r}\\
    &=\frac{2C_{2r}}{\sqrt{B}}\cdot\left\Vert \left(\mathbb{E}\left[
    \sup_x\left\|\nabla\hat{h}(Z^1)(x)\right\|_2^2\cdot \mathbbm{1}_{A_1}\Big|Z_{(l)}\right]\right)^{1/2} \right\Vert_{2r}.
\end{align}

Further, we rewrite this term as

\begin{align}
    &\frac{2C_{2r}}{\sqrt{B}}\cdot\left\Vert \left(\mathbb{E}\left[
    \sup_x\left\|\nabla\hat{h}(Z^1)(x)\right\|_2^2\cdot \mathbbm{1}_{A_1}\Big|Z_{(l)}\right]\right)^{1/2} \right\Vert_{2r}\\
    &=\frac{2C_{2r}}{\sqrt{B}}\cdot\left\Vert \left(\ 
    \mathbb{P}\left(A_1\Big|Z_{(l)}\right)
    \cdot\mathbb{E}\left[
    \sup_x\left\|\nabla\hat{h}(Z^1)(x)\right\|_2^2\Big|Z_{(l)}, A_1\right]\right)^{1/2} \right\Vert_{2r}\\
    &\hspace{3cm}\mbox{since $\nabla\hat{h}(Z^1)(x)=0$ for all $x$ on $A_1^c$}\\
    &=\frac{2C_{2r}}{\sqrt{B}}\cdot\left\Vert \left(\ 
    \mathbb{P}\left(A_1\right)
    \cdot\mathbb{E}\left[
    \sup_x\left\|\nabla\hat{h}(Z^1)(x)\right\|_2^2\Big|Z_{(l)}, A_1\right]\right)^{1/2} \right\Vert_{2r}\\
    &\hspace{3cm}\mbox{since $A_1$ is independent of $Z_{(l)}$}\\
    &=\frac{2C_{2r}}{\sqrt{B}}\cdot\left(\mathbb{P}\left(A_1\right) \right)^{1/2} \left\Vert \left(\ 
    \mathbb{E}\left[
    \sup_x\left\|\nabla\hat{h}(Z^1)(x)\right\|_2^2\Big|Z_{(l)}, A_1\right]\right)^{1/2} \right\Vert_{2r}.
\end{align}

By Jensen's inequality and tower law, we have
\begin{align}
    &\frac{2C_{2r}}{\sqrt{B}}\cdot\left(\mathbb{P}\left(A_1\right) \right)^{1/2} \left\Vert \left(\ 
    \mathbb{E}\left[
    \sup_x\left\|\nabla\hat{h}(Z^1)(x)\right\|_2^2\Big|Z_{(l)}, A_1\right]\right)^{1/2} \right\Vert_{2r}\\
    &\leq \frac{2C_{2r}}{\sqrt{B}}\cdot\left(\mathbb{P}\left(A_1\right) \right)^{1/2}
    \left\Vert 
    \sup_x\left\|\nabla\hat{h}(Z^1)(x)\right\|_2 \right\Vert_{2r}\\
    &\leq \frac{2C_{2r}}{\sqrt{B}}\cdot\left(\frac{m}{n} \right)^{1/2}\cdot 2C.
\end{align}

Similarly, by replacing the powers of $2$ and $1/2$ with $2r$ and $1/2r$, we can show that 
\begin{align}
    &\frac{C_{2r}}{\sqrt{B}}\cdot\left\Vert B^{\frac{1}{2r}-\frac{1}{2}}\left( \mathbb{E}\left[R_1^{2r}\Big|Z_{(l)}\right] \right)^{1/2r} \right\Vert_{2r}\\
    &\leq 2C_{2r} \cdot B^{\frac{1}{2r}-1}\left(\frac{m}{n} \right)^{1/2r}\cdot 2C.
\end{align}

Altogether, we have
\begin{align}
    &\max_{l\le n}\left\|\sup_x\|\hat g(x)-\hat g^{(-l)}(x)\|_2\right\|_{2r}=\left\|\sup_x\|\hat g(x)-\hat g^{(-l)}(x)\|_2\right\|_{2r}\\
    &\leq \frac{2C_{2r}}{\sqrt{B}}\cdot\left(\frac{m}{n} \right)^{1/2}\cdot 2C+2C_{2r} \cdot B^{\frac{1}{2r}-1}\left(\frac{m}{n} \right)^{1/2r}\cdot 2C+2C\cdot\frac{m}{n}.
\end{align}

\noindent Assuming 
$$m=o(\sqrt{n})$$
and
$$B>>m^{\frac{1}{2r-1}}\cdot n^{\frac{r-1}{2r-1}},$$
this upper bound is of order $o(n^{-1/2})$:
$$\max_{l\le n}\left\|\sup_x\|\hat g(x)-\hat g^{(-l)}(x)\|_2\right\|_{2r}=o(n^{-1/2}).$$

\end{proof}

\section{Proof of Lemma~\ref{lem: knn}}
\begin{proof}
Lemma~\ref{lem: knn} follows from Theorem 3 of \cite{khosravi2019non} by taking
$\psi(Z;\theta):=Y-\theta(x).$
\end{proof}

\section{Establishing \texorpdfstring{$L^{2r}$}{L\^2r}-Continuity and Mean-Squared-Continuity for Examples~\ref{sleep2}, \ref{sleep3}, and \ref{sleep1}}

Recall that
\begin{align}
    m(Z;\theta,g) = a(Z;g)\theta + \nu(Z;g).
\end{align}

We will establish $L^{2r}$-continuity and mean-squared-continuity for Examples~\ref{sleep2}, \ref{sleep3}, and \ref{sleep1} in this section.
\subsection{Establishing for Example~\ref{sleep2}}
For this example, we have
$$a(Z;g)=(T-p(X))(T-p(X))'$$
$$\nu(Z;g)=(Y-q(X))(T-p(X))'.$$

Let $v\ge 1$ be the constant such that $1 =\frac{1}{r}+\frac{1}{v}$. Denote $T_{i},i\in[p]$ as the $i$-th coordinate of $T$. Denote $Y_{i},i\in[p]$ as the $i$-th coordinate of $Y$. For any function $p$ denote $p_{i}(X),i\in[p]$ as the $i$-th coordinate of $p(X)$. We will show that subject to
$$\left\|T_i\right\|_{2v}, \left\|\hat p_i(X)\right\|_{2v}, \left\|\hat p_i(X_l)\right\|_{2v},\left\|Y\right\|_{2v},\left\|\hat q(X)\right\|_{2v}$$

being finite for all $i\in[p]$
then $L^{2r}$-continuity conditions hold. 

Moreover, we will show that if further
$$\left\|T_i\right\|_\infty, \left\| p_i(X)\right\|_\infty, \left\| p_i'(X)\right\|_\infty, \left\|Y\right\|_\infty, \left\| q'(X)\right\|_\infty$$
are finite for all $i\in[p]$ then mean-squared-continuity conditions hold with $q=2$.

We will illustrate for function $a$ and for function $\nu$ separately.


\underline{For function $a$:}\\
We first verify $L^{2r}$-continuity for function $a$. We have that for any $i,j\in[p]$
\begin{align}
&a_{i,j}(Z;\hat g)-a_{i,j}(Z;\hat g^{(-l)})\\
&=\left(\hat p_i^{(-l)}(X)-\hat p_i(X)\right)T_{j}+T_{i}\left(\hat p^{(-l)}_j(X)-\hat p_j(X)\right)\\
&+\frac{1}{2}\left(\hat p_i^{(-l)}(X)-\hat p_i(X)\right)\left(\hat p_j^{(-l)}(X)+\hat p_j(X)\right)\\
&+\frac{1}{2}\left(\hat p^{(-l)}_i(X)+\hat p_i(X)\right)\left(\hat p_j^{(-l)}(X)-\hat p_j(X)\right).
\end{align}

Then by triangle inequality and H\"older's inequality, we obtain
\begin{align}
&\left\|a_{i,j}(Z;\hat g)-a_{i,j}(Z;\hat g^{(-l)})\right\|_{2}\\
&\leq\left\|\sup_x\left|\hat p^{(-l)}_i(x)-\hat p_i(x)\right|\cdot\left|T_j\right|\right\|_2+\left\|\sup_x\left|\hat p^{(-l)}_j(x)-\hat p_j(x)\right|\cdot\left|T_i\right|\right\|_2\\
&+\frac{1}{2}\left\|\sup_x\left|\hat p^{(-l)}_i(x)-\hat p_i(x)\right|\cdot\left(\hat p^{(-l)}_j(X)+\hat p_j(X)\right)\right\|_2\\
&+\frac{1}{2}\left\|\sup_x\left|\hat p^{(-l)}_j(x)-\hat p_j(x)\right|\cdot\left(\hat p^{(-l)}_i(X)+\hat p_i(X)\right)\right\|_2\\
&\le \left\|\sup_x\left|\hat p^{(-l)}_i(x)-\hat p_i(x)\right|\right\|_{2r}\cdot\left\|T_j\right\|_{2v}+\left\|\sup_x\left|\hat p^{(-l)}_j(x)-\hat p_j(x)\right|\right\|_{2r}\cdot\left\|T_i\right\|_{2v}\\
&+\frac{1}{2}\left\|\sup_x\left|\hat p^{(-l)}_i(x)-\hat p_i(x)\right|\right\|_{2r}\cdot\left\|\hat p^{(-l)}_j(X)+\hat p_j(X)\right\|_{2v}\\
&+\frac{1}{2}\left\|\sup_x\left|\hat p^{(-l)}_j(x)-\hat p_j(x)\right|\right\|_{2r}\cdot\left\|\hat p^{(-l)}_i(X)+\hat p_i(X)\right\|_{2v}\\
&\overset{(a)}{\le} L_1\cdot\left\|\sup_x\left\|\hat g^{(-l)}(x)-\hat g(x)\right\|_2\right\|_{2r}
\end{align}
where 
\begin{align}
    L_1&:=\left\|T_j\right\|_{2v}+\left\|T_i\right\|_{2v}+\frac{1}{2}\left\|\hat p^{(-l)}_i(X)\right\|_{2v} +\frac{1}{2}\left\|\hat p_i(X)\right\|_{2v}+\frac{1}{2}\left\|\hat p^{(-l)}_j(X)\right\|_{2v} +\frac{1}{2}\left\|\hat p_j(X)\right\|_{2v}\\
    &\overset{(b)}{=}\left\|T_j\right\|_{2v}+\left\|T_i\right\|_{2v}+\left\|\hat p_i(X)\right\|_{2v}+\left\|\hat p_j(X)\right\|_{2v}.
\end{align}

Here for (a) we have used the triangle inequality, and for (b) we have used the fact that $Z_l$ and $\tilde{Z}_l$ have the same distribution.

By replacing $Z$ with $Z_l$, we can similarly show that
\begin{align}
&\left\|a_{i,j}(Z_l;\hat g)-a_{i,j}(Z_l;\hat g^{(-l)})\right\|_{2}\\
& \le L_2\cdot \left\|\sup_x\left\|\hat g^{(-l)}(x)-\hat g(x)\right\|_2\right\|_{2r}
\end{align}
where $$L_2:=\left\|T_j\right\|_{2v}+\left\|T_i\right\|_{2v}+\frac{1}{2}\left\|\hat p_i(X)\right\|_{2v} +\frac{1}{2}\left\|\hat p_i(X_l)\right\|_{2v}+\frac{1}{2}\left\|\hat p_j(X)\right\|_{2v} +\frac{1}{2}\left\|\hat p_j(X_l)\right\|_{2v}.$$

Hence, the $L^{2r}$-continuity conditions hold for function $a$ provided that all the aforementioned $L^{2v}$-norm quantities are finite.

Now we check the mean-squared-continuity conditions for function $a.$

For any $g,g'$, any $i,j\in[p]$, we have
\begin{align}
&a_{i,j}(Z;g)-a_{i,j}(Z; g')\\
&=\left(p_i'(X)- p_i(X)\right)T_{j}+T_{i}\left( p_j'(X)-p_j(X)\right)\\
&+\frac{1}{2}\left( p_i'(X)- p_i(X)\right)\left( p_j'(X)+ p_j(X)\right)\\
&+\frac{1}{2}\left( p_i'(X)+ p_i(X)\right)\left( p_j'(X)- p_j(X)\right).
\end{align}

Then by triangle inequality and H\"older's inequality, we obtain
\begin{align}
&\left\|a_{i,j}(Z;g)-a_{i,j}(Z; g')\right\|_{2}\\
&\leq\left\|p_i'(X)- p_i(X)\right\|_2\cdot\|T_j\|_\infty+\left\|p_j'(X)- p_j(X)\right\|_2\cdot\|T_i\|_\infty\\
&+\frac{1}{2}\left\|p_i'(X)- p_i(X)\right\|_2\cdot\|p_j'(X)+ p_j(X)\|_\infty+\frac{1}{2}\left\|p_j'(X)- p_j(X)\right\|_2\cdot\|p_i'(X)+ p_i(X)\|_\infty\\
&\le L_3\cdot \|g-g'\|_2
\end{align}
where
$$L_3:=\|T_j\|_\infty+\|T_i\|_\infty+\frac{1}{2}\|p_j'(X)\|_\infty+\frac{1}{2}\|p_j(X)\|_\infty+\frac{1}{2}\|p_i'(X)\|_\infty+\frac{1}{2}\|p_i(X)\|_\infty.$$

Provided all these $L^\infty$-norm quantities are finite, mean-squared-continuity conditions hold for function $a$ with $q=2.$

\underline{For function $\nu$:}\\
We have
\begin{align}
&\nu_{i}(Z;\hat g)-\nu_{i}(Z;\hat g^{(-l)})=\left(\hat q^{(-l)}(X)-\hat q(X)\right)T_i+Y\left(\hat p_i^{(-l)}(X)-\hat p_i(X)\right)\\
&-\left(\hat q^{(-l)}(X)-\hat q(X)\right)\hat p_i(X)-\hat q^{(-l)}(X)\left(\hat p_i^{(-l)}(X)-\hat p_i(X)\right).
\end{align}

Hence, by triangle inequality and H\"older's inequality, we similarly obtain that

\begin{align}
&\left\|\nu_{i}(Z;\hat g)-\nu_{i}(Z;\hat g^{(-l)})\right\|_2\\
&\le\left\|\sup_x\left|\hat q^{(-l)}(x)-\hat q(x)\right|\right\|_{2r}\cdot \left\{ \left\|T_i\right\|_{2v}+\left\|\hat p_i(X)\right\|_{2v} \right\}\\
&+\left\|\sup_x\left|\hat p^{(-l)}_i(x)-\hat p_i(x)\right|\right\|_{2r}\cdot \left\{ \left\|Y\right\|_{2v}+\left\|\hat q^{(-l)}(X)\right\|_{2v} \right\}\\
&\le \left\|\sup_x\left\|\hat g^{(-l)}(x)-\hat g(x)\right\|_2\right\|_{2r}\cdot\left\{\left\|T_i\right\|_{2v}+\left\|\hat p_i(X)\right\|_{2v} + \left\|Y\right\|_{2v}+\left\|\hat q(X)\right\|_{2v} \right\}.
\end{align}

By replacing $Z$ with $Z_l$, we can similarly show that

\begin{align}
&\left\|\nu_{i}(Z_l;\hat g)-\nu_{i}(Z_l;\hat g^{(-l)})\right\|_2\\
&\le \left\|\sup_x\left\|\hat g^{(-l)}(x)-\hat g(x)\right\|_2\right\|_{2r}\cdot\left\{\left\|T_i\right\|_{2v}+\left\|\hat p_i(X_l)\right\|_{2v} + \left\|Y\right\|_{2v}+\left\|\hat q(X)\right\|_{2v} \right\}.
\end{align}

Therefore, the $L^{2r}$-continuity conditions hold for function $\nu$ provided that all these $L^{2v}$-norm quantities are finite.

As for mean-squared-continuity, note that we have that 
\begin{align}
&\nu_{i}(Z; g)-\nu_{i}(Z;g')\\
&=\left(q'(X)- q(X)\right)T_i+Y\left(p_i'(X)- p_i(X)\right)\\
&-\left(q'(X)-q(X)\right) p_i(X)-q'(X)\left(p_i'(X)- p_i(X)\right).
\end{align}

Hence, by triangle inequality and H\"older's inequality, we derive that
\begin{align}
    &\| \nu_{i}(Z; g)-\nu_{i}(Z;g') \|_2\\
    &\le \|q'(X)- q(X)\|_2\cdot \| T_i\|_\infty+\|p_i'(X)- p_i(X)\|_2\cdot \| Y\|_\infty\\
    &+\|q'(X)- q(X)\|_2\cdot \| p_i(X)\|_\infty+\|p_i'(X)- p_i(X)\|_2\cdot \| q'(X)\|_\infty\\
    &\le \|g-g'\|_2\cdot\left\{\| T_i\|_\infty+ \| Y\|_\infty+ \| p_i(X)\|_\infty+\| q'(X)\|_\infty \right\}.
\end{align}

Provided all these $L^\infty$-norm quantities are finite, mean-squared-continuity conditions hold for function $\nu$ with $q=2.$

\subsection{Establishing for Example~\ref{sleep3}}

For this example, we have
$$a(Z;g)=(Z-r(X))(T-p(X))'$$
$$\nu(Z;g)=(Y-q(X))(Z-r(X))'.$$

Denote $Z_{i},i\in[p]$ as the $i$-th coordinate of $Z$. Denote $ r_{i}(X),i\in[p]$ as the $i$-th coordinate of $ r(X)$. Analogously to Example~\ref{sleep2}, replacing all functions $p_i$ and their estimates with $r_i$ and their corresponding estimates, replacing all $T_i$ with $Z_i$ in the analysis of $a_{i,j}$ and $\nu_i$, we can show that subject to
$$\left\|T_i\right\|_{2v},\left\|Z_i\right\|_{2v}, \left\|\hat p_i(X)\right\|_{2v}, \left\|\hat p_i(X_l)\right\|_{2v},\left\|\hat r_i(X)\right\|_{2v}, \left\|\hat r_i(X_l)\right\|_{2v},\left\|Y\right\|_{2v},\left\|\hat q(X)\right\|_{2v}$$
being finite for all $i\in[p]$
then $L^{2r}$-continuity conditions hold. Moreover,
if further
$$\left\|T_i\right\|_\infty,\left\|Z_i\right\|_\infty, \left\| p_i(X)\right\|_\infty, \left\| p_i'(X)\right\|_\infty,\left\| r_i(X)\right\|_\infty, \left\| r_i'(X)\right\|_\infty, \left\|Y\right\|_\infty, \left\| q'(X)\right\|_\infty$$
are finite for all $i\in[p]$ then mean-squared-continuity conditions hold with $q=2$.

\subsection{Establishing for Example~\ref{sleep1}}
For this example, we have
$$a(Z;g)\equiv 1$$
$$\nu(Z;g)=-m_b(Z;q)-\mu(T,X)(Y-q(T,X)).$$

The $L^{2r}$-continuity and mean-squared-continuity conditions trivially hold for function $a$. For function $\nu$, we have

\begin{align}
&\nu(Z;\hat g)-\nu(Z;\hat g^{(-l)})=(m_b(Z;\hat q^{(-l)})-m_b(Z;\hat q))\\
&+\hat \mu^{(-l)}(T,X)(\hat q(T,X)-\hat q^{(-l)}(T,X))-(Y-\hat q(T,X))(\hat \mu(T,X)-\hat \mu^{(-l)}(T,X))
.
\end{align}

Hence, by triangle inequality and H\"older's inequality we obtain
\begin{align}
    &\|\nu(Z;\hat g)-\nu(Z;\hat g^{(-l)})\|_2\\
    &\leq \| m_b(Z;\hat q^{(-l)})-m_b(Z;\hat q)\|_2\\
    &+\|\hat \mu^{(-l)}(T,X)(\hat q(T,X)-\hat q^{(-l)}(T,X))\|_2+\| (Y-\hat q(T,X))(\hat \mu(T,X)-\hat \mu^{(-l)}(T,X))\|_2\\
    &\le \| m_b(Z;\hat q^{(-l)})-m_b(Z;\hat q)\|_2\\
    &+\|\hat \mu^{(-l)}(T,X)\|_{2v}\cdot\|\hat q(T,X)-\hat q^{(-l)}(T,X)\|_{2r}\\
    &+\| Y-\hat q(T,X)\|_{2v}\cdot\|\hat \mu(T,X)-\hat \mu^{(-l)}(T,X)\|_{2r}\\
    &\le \| m_b(Z;\hat q^{(-l)})-m_b(Z;\hat q)\|_2\\
    &+\|\hat \mu(T,X)\|_{2v}\cdot\|\sup_{t,x}|\hat q(t,x)-\hat q^{(-l)}(t,x)|\|_{2r}
    \\
    &+(\| Y\|_{2v}+\|\hat q(T,X)\|_{2v})\|\sup_{t,x}|\hat \mu(t,x)-\hat \mu^{(-l)}(t,x)|\|_{2r}.
\end{align}

Since $m_b$ is a linear functional, there exists $L_m>0$ such that
$$\| m_b(Z;\hat q^{(-l)})-m_b(Z;\hat q)\|_2\leq L_m\cdot\|\sup_{t,x}|\hat q(t,x)-\hat q^{(-l)}(t,x)|\|_{2r}.$$

Hence, we have
\begin{align}
    &\|\nu(Z;\hat g)-\nu(Z;\hat g^{(-l)})\|_2\\
    &\leq \|\sup_{t,x}|\hat g(t,x)-\hat g^{(-l)}(t,x)|\|_{2r}\cdot \left\{L_m+\|\hat \mu(T,X)\|_{2v}+\| Y\|_{2v}+\|\hat q(T,X)\|_{2v}\right\}.
\end{align}

Analogously, by replacing $Z$ with $Z_l$, we can show that
\begin{align}
    &\|\nu(Z_l;\hat g)-\nu(Z_l;\hat g^{(-l)})\|_2\\
    &\leq \|\sup_{t,x}|\hat g(t,x)-\hat g^{(-l)}(t,x)|\|_{2r}\cdot \left\{\tilde{L}_m+\|\hat \mu(T,X)\|_{2v}+\| Y\|_{2v}+\|\hat q(T_l,X_l)\|_{2v}\right\}
\end{align}
for some constant $\tilde{L}_m>0$.

Therefore, subject to
$$L_m,\tilde{L}_m,\|\hat \mu(T,X)\|_{2v},\| Y\|_{2v},\|\hat q(T,X)\|_{2v}, \|\hat q(T_l,X_l)\|_{2v}$$
being finite for all $i\in[p]$
then $L^{2r}$-continuity conditions also hold for function $\nu$.

Moreover, we have
\begin{align}
    &\|\nu(Z;g)-\nu(Z; g')\|_2\\
    &\leq \| m_b(Z;q')-m_b(Z; q)\|_2\\
    &+\|\mu'(T,X)(q(T,X)- q'(T,X))\|_2+\| (Y- q(T,X))( \mu(T,X)- \mu'(T,X))\|_2\\
    &\le \|g-g'\|_2\cdot\left\{L_b+\| \mu'(T,X)\|_\infty+ \| Y\|_\infty+\| q(T,X)\|_\infty \right\}
\end{align}
where since $m_b$ is a linear functional, there exists $L_b>0$ such that
$$\| m_b(Z;q')-m_b(Z;q)\|_2\leq L_b\cdot\|q(t,x)-q'(t,x)\|_2.$$

Provided that $$L_b,\| \mu'(T,X)\|_\infty, \| Y\|_\infty,\| q(T,X)\|_\infty $$
are finite, mean-squared-continuity conditions hold for function $\nu$ with $q=2.$

\section{Extension to Nonlinear Moments}\label{sec: nonlinear}

In this section, we extend our results to the case where the moment function $m(Z;\theta,g)$ is not necessarily linear in the target parameter $\theta.$ For simplicity, we assume that the nuisance estimator $\hat g$ is symmetric in each of the training data points $Z_1,\ldots,Z_n.$ Moreover, we will denote with $Z$ a fresh random draw from the distribution.

We introduce some notation.
We denote with $\|\cdot\|_{2,2}$ the norm of a random vector defined as: $\|Z\|_{2,2}=\sqrt{\E\left[\sum_{i} Z_i^2\right]}$, which can also be thought as taking the $L_2$ norm of each coordinate and and then taking the $\ell_2$ vector norm of this vector, or equivalently taking the $L_2$ norm of the random variable defined as the $\ell_2$ norm of the random vector. For clarity, for any random vector $Z$ we will denote with $\|Z\|_{v,2}$ the random variable that corresponds to the $\ell_2$ vector norm of the random vector, i.e. $\|Z\|_{v,2}=\sqrt{\sum_i Z_i^2}$. For any random variable $V$, denote with $\|V\|_{2}$ its $\ell_2$ norm $\sqrt{\E[V^2]}$. Note that for any random vector $Z$, we have $\|Z\|_{2,2} = \left\| \|Z\|_{v,2} \right\|_2$.

Firstly, we establish a consistency lemma for $\hat\theta$. We note that in the linear moment case, such a separate proof of consistency was not required and a single step proof of asymptotic normality was feasible due to linearity. In the non-linear case, as is typical for moment based estimators, we first need to show that the estimate will eventually lie in a small ball around $\theta_0$, and then argue normality. This is what the consistency lemma achieves.

\begin{lemma}[Consistency]\label{lem:nonlinear}
Assume that 
\begin{enumerate}
    \item The parameter space $\Theta\subset\mathbb{R}^p$ for target parameter $\theta$ is compact.
    \item $\theta_0$ is the unique solution of $\theta$ to the equation $M(\theta, g_0)=0.$
    \item The moment function $m(z;\theta, g)$ is uniformly continuous in $\theta$ over all $\Theta$ and a suficciently small $\ell_2$ ball $B_2(g_0)\subseteq \mathcal{G}$ around $g_0$. That is, $\forall \epsilon>0$, $ \exists\delta>0$ such that for any $\tilde{\theta}_1,\tilde{\theta}_2\in\Theta$ with $\|\tilde{\theta}_1-\tilde{\theta}_2\|_2<\delta$, $\forall g\in B_2(g_0)$, $\forall z$, we have
    $$\|m(z;\tilde{\theta}_1,g)-m(z;\tilde{\theta}_2,g)\|_2<\epsilon.$$
    \item The moment function $m(Z;\theta, g)$ is mean-squared-continuous in $g$, uniformly in $\theta$, i.e. $\exists L>0$ and $q>0$ such that:  
    $$\max_{\theta\in\Theta} \| m(Z;\theta, g_1)-m(Z;\theta, g_2)\|_{2,2}\le L\cdot \|g_1(Z)-g_2(Z)\|_{2,2}^{q}.$$
    \item Estimator $\hat g$ of the nuisance function is consistent: as $n\to\infty$
     $$\left\|\hat g(Z)-g_0(Z)\right\|_{2,2}
    =o(1).$$
    \item The moment function $m$ and estimator $\hat{g}$ satisfies the following $o(1)$ leave-one-out stability condition:
    $$\max_{\theta\in\Theta} \|m(Z_1;\theta,\hat g)-m(Z_1;\theta, \hat g^{(-1)})\|_{2,2}=o(1)$$
    as $n\to\infty.$
\end{enumerate}
Then any estimator $\hat{\theta}$ that satisfies that $M_n(\hat{\theta}, \hat{g})=o_p(1)$, also satisfies that $\hat{\theta} \overset{p}{\to} \theta_0$.
\end{lemma}

\begin{proof}

Fix any $\epsilon>0$. Since $m(z;\theta, g)$ is uniformly continuous in $\theta\in\Theta$ for a sufficiently small $\ell_2$-ball $B_2(g_0)$ around $g_0$, we have that, $\exists\delta>0$ such that for any $\tilde{\theta}_1,\tilde{\theta}_2\in\Theta$ with $\|\tilde{\theta}_1-\tilde{\theta}_2\|_2<\delta$, $\forall g\in B_2(g_0)$, $\forall z$, we have
$$\|m(z;\tilde{\theta}_1,g)-m(z;\tilde{\theta}_2,g)\|_2<\epsilon/6.$$
Then
$$\|M(\tilde{\theta}_1,g)-M(\tilde{\theta}_2,g)\|_2\le \mathbb{E}\left[ \|m(Z;\tilde{\theta}_1,g)-m(Z;\tilde{\theta}_2,g)\|_{v,2}\right] \le \epsilon/6.$$
(that is, $M(\cdot,g_0)$ is uniformly continuous) and
\begin{align}
    \|M_n(\tilde{\theta}_1,g)-M_n(\tilde{\theta}_2,g)\|_{v,2}\le \frac{1}{n}\sum_{i=1}^n\left\|m(Z_i;\tilde{\theta}_1,g)-m(Z_i;\tilde{\theta}_2,g)\right\|_{v,2}\le \epsilon/6.
\end{align}

\noindent Since the parameter space $\Theta$ is compact, there exist $\theta_j, j=1,\ldots,J$ such that
$$\Theta\subset\cup_{j=1}^J\mathcal{B}(\theta_j,\delta).$$

By Law of Large Numbers, we have $\forall j,$ as $n\to\infty$
$$M_n(\theta_j,g_0)-M(\theta_j,g_0)\overset{p}{\to}0.$$
Hence, $\forall \eta>0$, for every $j$ there exists $n_j$ such that $\forall n>n_j$
$$\mathbb{P}\left(\|M_n(\theta_j,g_0)-M(\theta_j,g_0)\|_{v,2}>\frac{\epsilon}{3}\right)<\frac{\eta}{3J}.$$

Then $\forall n>\max_j n_j$, we have
$$\mathbb{P}\left(\max_j\|M_n(\theta_j,g_0)-M(\theta_j,g_0)\|_{v,2}>\frac{\epsilon}{3}\right)<\frac{\eta}{3}.$$

Moreover, for any $j\in J$, we have that:
\begin{align}
    &~\left\|M_n(\theta_j, \hat g)-M_n(\theta_j, g_0)\right\|_{2,2}\\
    \le~& \max_\theta\left\|M_n(\theta, \hat g)-M_n(\theta, g_0)\right\|_{2,2}\\
    =~&  \max_{\theta}\left\|\frac{1}{n}\sum_{i=1}^n \left\{m(Z_i;\theta, \hat g)-m(Z_i;\theta, g_0)\right\}\right\|_{2,2}\\
    \le~& \frac{1}{n}\sum_{i=1}^n\max_\theta \left\|m(Z_i;\theta,\hat g)-m(Z_i;\theta, g_0)\right\|_{2,2} \tag{triangle inequality}\\
    =~&  \max_\theta \left\|m(Z_1;\theta,\hat g)-m(Z_1;\theta, g_0)\right\|_{2,2} \tag{symmetry of estimator}\\
    \leq~& \max_\theta \left\| m(Z_1;\theta,\hat g)-m(Z_1;\theta, \hat{g}^{(-1)})\right\|_{2,2} + \max_\theta \left\|m(Z_1;\theta,\hat g^{(-1)})-m(Z_1;\theta, g_0)\right\|_{2,2}\\
    \leq~&  \max_\theta \left\|m(Z_1;\theta,\hat g)-m(Z_1;\theta, \hat{g}^{(-1)})\right\|_{2,2} + \max_\theta \left\| m(Z;\theta,\hat g^{(-1)})-m(Z;\theta, g_0)\right\|_{2,2}\\
    \leq~&  \max_\theta \left\|m(Z_1;\theta,\hat g)-m(Z_1;\theta, \hat{g}^{(-1)})\right\|_{2,2} + L\cdot\left\|\hat g^{(-1)}(Z)- g_0(Z)\right\|_{2,2}^q\\
    =~& o(1).
\end{align}
Thus we have that $M_n(\theta_j, \hat g)-M_n(\theta_j, g_0)=o_p(1)$. Which means that $\forall \eta>0$, there exists $n_j$ such that for every $n>n_j$:
\begin{align}
    \mathbb{P}\left(\left\|M_n(\theta_j, \hat g)-M_n(\theta_j, g_0)\right\|_{v,2} > \frac{\epsilon}{3}\right) < \frac{\eta}{3J}.
\end{align}
Then $\forall n>\max_{j\in J}n_j$:
\begin{align}
    \mathbb{P}\left(\max_{j\in [J]} \left\|M_n(\theta_j, \hat g)-M_n(\theta_j, g_0)\right\|_{v,2} > \frac{\epsilon}{3}\right) < \frac{\eta}{3}.
\end{align}

Now $\forall\theta\in\Theta,$ since $\Theta\subset\cup_{j=1}^J\mathcal{B}(\theta_j,\delta)$, there exists $k\in\{1,\ldots,J\}$ such that
$\|\theta-\theta_k\|_2<\delta$. Then for $n$ sufficiently large, such that $\hat{g}\in B_2(g_0)$:
\begin{align}
    &\|M_n(\theta,\hat{g})-M(\theta,g_0)\|_{v,2}\\
    \le~& \|M_n(\theta_k, \hat{g})-M(\theta_k,g_0)\|_{v,2}+\|M_n(\theta, \hat{g})-M_n(\theta_k,\hat{g})\|_{v,2}+\|M(\theta,g_0)-M(\theta_k,g_0)\|_{v,2}\\
    \le~& \max_j\|M_n(\theta_j,\hat{g})-M(\theta_j,g_0)\|_{v,2}  + 2\epsilon/6.
\end{align}
Hence, we obtain
\begin{align}
    \mathbb{P}\left(\max_\theta\|M_n(\theta, \hat{g})-M(\theta,g_0)\|_{v,2}>\epsilon\right)
    \le~& \mathbb{P}\left(\max_j\|M_n(\theta_j,\hat{g})-M(\theta_j,g_0)\|_{v,2}>\frac{2\epsilon}{3}\right).
\end{align}
Moreover, note that by the triangle inequality:
\begin{align}
    &\max_j\|M_n(\theta_j,\hat{g})-M(\theta_j,g_0)\|_{v,2} \\
    \leq~& \max_j\|M_n(\theta_j,\hat{g})-M_n(\theta_j,g_0)\|_{v,2} + \max_j\|M_n(\theta_j, g_0)-M(\theta_j,g_0)\|_{v,2}.
\end{align}
Thus:
\begin{align}
    &\mathbb{P}\left(\max_j\|M_n(\theta_j,\hat{g})-M(\theta_j,g_0)\|_{v,2}>\frac{2\epsilon}{3}\right)\\
    \leq~& \mathbb{P}\left(\max_j\|M_n(\theta_j,\hat{g})-M_n(\theta_j,g_0)\|_{v,2}>\frac{\epsilon}{3}\right) + \mathbb{P}\left(\max_j\|M_n(\theta_j, g_0)-M(\theta_j,g_0)\|_{v,2}>\frac{\epsilon}{3}\right)\\
    \leq~& \frac{2\eta}{3} \leq \eta.
\end{align}
And we conclude that:
\begin{align}
    \mathbb{P}\left(\max_\theta\|M_n(\theta, \hat{g})-M(\theta,g_0)\|_{v,2}>\epsilon\right)
    \le~& \eta.
\end{align}
\noindent This shows that
$$\max_{\theta\in \Theta}\|M_n(\theta,\hat{g})-M(\theta,g_0)\|_{v,2}\overset{p}{\to}0$$
as $n\to\infty.$

In particular, this implies that
$$\|M_n(\hat\theta, \hat{g})-M(\hat\theta,g_0)\|_{v,2}\overset{p}{\to}0$$
as $n\to\infty.$

Hence, by triangle inequality and the fact that $M_n(\hat\theta,\hat g)=o_p(1)$, we obtain
$$\|M(\hat{\theta},g_0)\|_{v,2} \leq \|M_n(\hat\theta, \hat g)\|_{v,2} + \|M_n(\hat\theta, \hat g) - M(\hat{\theta}, g_0)\|_{v,2} =o_p(1).$$

Hence, $M(\hat\theta, g_0) = o_p(1)$.

It remains to show that $$\hat\theta\overset{p}{\to}\theta_0$$
as $n\to\infty.$

To achieve this, again fix any $\epsilon>0.$
Then since $\Theta$ is compact, $B(\theta_0,\epsilon)^c$ is also compact as a closed subset of $\Theta$. By continuity of $\theta\mapsto \|M(\theta,g_0)\|_2$ and the fact that $\theta_0$ is the unique solution to $M(\theta_0,g_0)=0$, we must have that $\|M(\theta,g_0)\|_{v,2}$ is bounded away from zero on $B(\theta_0,\epsilon)^c$. That is, $\exists \eta>0$ such that for any $\theta$ with $\|\theta-\theta_0\|_{v,2}\geq \epsilon$,
$$\|M(\theta,g_0)\|_{v,2}>\eta.$$

Then since $M(\hat{\theta},g_0)=o_p(1)$, there exists $N\in\mathbb{N}$ such that $\forall n>N$ $$\mathbb{P}\left(\|M(\hat\theta,g_0)\|_{v,2}>\eta\right)<\epsilon.$$

Then $\forall n>N$
$$\mathbb{P}\left(\|\hat\theta-\theta_0\|_{v,2}\geq \epsilon\right)\le \mathbb{P}\left(\|M(\hat\theta,g_0)\|_{v,2}>\eta\right)<\epsilon.$$
This establishes consistency of $\hat\theta.$
\end{proof}

Now we extend Theorem~1 to nonlinear moments.

\begin{theorem}
Let $A(\theta,g):=\partial_{\theta} M(\theta, g)$ denote the Jacobian of the moment vector, with respect to $\theta$ and $H_i(\theta,g):=\partial_{\theta}^2 M_i(\theta,g)$ denote the Hessian of the $i$-th moment coordinate. Suppose that the moment $m$ is twice differentiable with $a(z; \theta,g):= \partial_{\theta} m(z;\theta, g)$ and $h_i(z;\theta, g):=\partial_\theta^2 m_i(z;\theta ,g)$. Let $A_n(\theta, g):=\frac{1}{n}\sum_{i=1}^n a(Z_i; \theta, g)$ and $H_{i,n}:=\frac{1}{n}\sum_{i=1}^n h(Z_i; \theta, g)$ denote the empirical counterparts of $A, H_i$.

Suppose that the nuisance estimate $\hat{g}\in \mcG$ satisfies:
\begin{equation}\label{cons}
\|\hat{g} - g_0\|^2_2~\defeq~ \E_X\left[\|\hat{g}(X) - g_0(X)\|_2^2\right] =o_p\left(n^{-1/2}\right). \tag{Consistency Rate}
\end{equation}
Suppose that the moment satisfies the Neyman orthogonality condition: for all $g\in \mcG$
\begin{align}
    D_g M(\theta_0, g_0)[g-g_0] ~\defeq~ \frac{\partial}{\partial t} M(\theta_0, g_0 + t\, (g-g_0))\big|_{t=0} ~=~& 0 \tag{Neyman Orthogonality}
\end{align}
and a second-order smoothness condition: for all $g\in \mcG$
\begin{align}
    D_{gg} M(\theta_0, g_0)[g-g_0] ~\defeq~ \frac{\partial^2}{\partial t^2} M(\theta_0, g_0 + t\, (g-g_0))\big|_{t=0} ~=~& O\left(\|g-g_0\|_2^2\right) \tag{Smoothness}
\end{align}
Assume that $A(\theta_0, g_0)^{-1}$ exists and that for any $g,g'\in \mcG$:
\begin{equation}
\|A(\theta_0, g) - A(\theta_0, g')\|_{op}=O\left(\|g-g'\|_2\right).    \tag{Lipschitz}
\end{equation}
Suppose that the moment $m$ and estimator $\hat{g}$ satisfy the stochastic equicontinuity conditions:
\begin{align}\label{eqn:equicont-nonlin}
\begin{aligned}
    \left\|A(\theta_0,\hat g)-A(\theta_0,g_0) -\left(A_n(\theta_0,\hat g)-A_n(\theta_0,g_0)\right)\right\|_{op}=o_p(1)\\
    \sqrt{n} \left\|M(\theta_0, \hat{g}) - M(\theta_0, g_0) - (M_n(\theta_0, \hat{g}) - M_n(\theta_0, g_0))\right\|_{2,2} =~& o_p(1)
\end{aligned} \tag{NonLin. Stoc. Equi.}
\end{align}
Suppose that the conditions in Lemma~\ref{lem:nonlinear} are true. Moreover, assume that for any $i,j\in [p]\times[p]$, the random variable $a_{i,j}(Z;\theta_0, g_0)$ has bounded variance and that $\|\theta_0\|_2=O(1)$. Suppose that $\exists$ open neighborhood $W$ of $g_0$ such that
\begin{align}
\sup_{\theta\in \Theta,g\in W, i\in [p]}\left\|H_i(\theta,g)\right\|_{op}, \left\|H_{i,n}(\theta,g)\right\|_{op}<\infty \tag{Bounded Hessian}
\end{align}
Let $\hat{\theta}$ denote any approximate solution to the plug-in empirical moment equation that satisfies $M_n(\hat{\theta}, \hat{g})=o_p(n^{-1/2})$. Then $\hat{\theta}$ is asymptotically normal:
\begin{equation}
    \sqrt{n} \left(\hat{\theta} - \theta_0\right) \xrightarrow{n\rightarrow \infty, d} N\left(0, A(\theta_0, g_0)^{-1} \E\left[m(Z; \theta_0, g_0)\, m(Z; \theta_0, g_0)^\top\right] A(\theta_0, g_0)^{-1}\right).
\end{equation}
\end{theorem}
\begin{proof}
Let $A_i(\theta, g)$ denote the $i$-th column of the Jacobian matrix. By Taylor's expansion, we have $\forall g$ and $i\in [p]$
\begin{align}
    M_i(\hat\theta,g)-M_i(\theta_0,g)
     =~& A_i(\theta_0, g)'(\hat\theta-\theta_0)
    +(\hat\theta-\theta_0)'H_i(\tilde{\theta}_i,g)(\hat\theta-\theta_0)
\end{align}
for some $\tilde{\theta}_i$ between $\hat\theta$ and $\theta_0$.\\

By Lemma~\ref{lem:nonlinear}, we know that $$\|\hat\theta-\theta_0\|_2=o_p(1).$$

Further, by the bounded Hessian condition and consistency of $\hat\theta$, we know that uniformly for $g\in W$
$$(\hat\theta-\theta_0)'H_i(\tilde{\theta}_i,g)(\hat\theta-\theta_0)=O_p\left(\|\hat\theta-\theta_0\|_2^2\right)=o_p(\|\hat\theta-\theta_0\|_2).$$

Moreover, for any $g$ with $\|g-g_0\|_2=o_p(1)$ we have
\begin{align}
    A(\theta_0,g)\cdot (\hat\theta-\theta_0)=~& A(\theta_0,g_0)\cdot (\hat\theta-\theta_0)
    +\left(A(\theta_0,g) - A(\theta_0,g_0) \right)\cdot (\hat\theta-\theta_0)\\
    =~& A(\theta_0,g_0)\cdot (\hat\theta-\theta_0)
    +O(\|g-g_0\|_2\ \|\hat\theta-\theta_0\|_2)\\
    =~& A(\theta_0,g_0)\cdot (\hat\theta-\theta_0)
    +o_p(\|\hat\theta-\theta_0\|_2),
\end{align}
where the second to last equality uses the Lipschitz condition.

By consistency, we have $\|\hat g-g_0\|_2=o_p(1)$ and that $\hat g\in W$ for $n$ large enough.

Hence, we obtain that for large enough $n$,
\begin{align}
    A(\theta_0,g_0)\cdot (\hat\theta-\theta_0)&= M(\hat\theta,\hat g)-M(\theta_0,\hat g) +o_p(\|\hat\theta-\theta_0\|_2)\\
    &= M(\hat\theta,\hat g)-M_n(\hat\theta,\hat g)+M(\theta_0,g_0)-M(\theta_0,\hat g) +M_n(\hat\theta,\hat g)+o_p(\|\hat\theta-\theta_0\|_2)\\
    &= M(\hat\theta,\hat g)-M_n(\hat\theta,\hat g)+M(\theta_0,g_0)-M(\theta_0,\hat g) +o_p(n^{-1/2}+\|\hat\theta-\theta_0\|_2),
\end{align}
where the last line follows because by definition $M_n(\hat\theta,\hat g)=o_p(n^{-1/2}).$

By Neyman orthogonality and the boundness condition on the second derivative of $M$ with respect to $g$, for any $g$ we have
\begin{align}
    M(\theta_0, g_0) - M(\theta_0, g) = D_g\, M(\theta_0, g_0) [g_0 - g] + O\left(\|g-g_0\|_2^2\right) = O\left(\|g-g_0\|_2^2\right).
\end{align}

Plugging in $g=\hat g$, noting that $\|g-g_0\|_2^2=o_p(n^{-1/2})$ we obtain
$$M(\theta_0, g_0) - M(\theta_0, \hat g)=o_p(n^{-1/2}).$$

Let $G_n(\theta,g):=M(\theta,g)-M_n(\theta,g).$ Thus we have that
$$A(\theta_0,g_0) \cdot (\hat\theta-\theta_0)=G_n(\hat\theta,\hat g)+o_p(n^{-1/2}+\|\hat\theta-\theta_0\|_2).$$

We decompose $G_n(\hat\theta,\hat g)$ into the following sum:
\begin{align}
    G_n(\hat\theta,\hat g)=G_n(\hat\theta,\hat g)-G_n(\theta_0,\hat g)+\left(G_n(\theta_0,\hat g)-G_n(\theta_0,g_0)\right) + G_n(\theta_0,g_0).
\end{align}

By Taylor's expansion, we have
\begin{align}
    G_n(\hat\theta,\hat g)-G_n(\theta_0,\hat g)&=M(\hat\theta,\hat g)-M(\theta_0,\hat g)-\left(M_n(\hat\theta,\hat g)-M_n(\theta_0,\hat g)\right)\\
    &=\left(A(\theta_0,\hat g)-A_n(\theta_0,\hat g)\right)'(\hat\theta-\theta_0)+o_p(\|\hat\theta-\theta_0\|_2)
\end{align}

Now
\begin{align}
    &\|A(\theta_0,\hat g)-A_n(\theta_0,\hat g)\|_{op}\\
    \le~& \|A(\theta_0, g_0)-A_n(\theta_0, g_0)\|_{op} +\left\|A(\theta_0,\hat g)-A(\theta_0,g_0) - \left(A_n(\theta_0,\hat g)-A_n(\theta_0,g_0)\right)\right\|_{op}\\
    =~& \|A(\theta_0, g_0)-A_n(\theta_0, g_0)\|_{op}+o_p(1),
\end{align}
where the last equality follows from the stochastic equicontinuity condition on the Jacobian.
Since $A(\theta_0, g_0)-A_n(\theta_0, g_0)$ is a mean zero empirical process with $\partial_{\theta} m(Z;\theta_0,g_0)$ having bounded variance, we have
$$\|A(\theta_0, g_0)-A_n(\theta_0, g_0)\|_{op}=o_p(1).$$
Hence,
\begin{align}
    G_n(\hat\theta,\hat g)-G_n(\theta_0,\hat g)=o_p(\|\hat\theta-\theta_0\|_2).
\end{align}

Moreover,
\begin{align}
    G_n(\theta_0,\hat g)-G_n(\theta_0,g_0)
    =~& M(\theta_0,\hat g)-M_n(\theta_0,\hat g) -\left(M(\theta_0,g_0)-M_n(\theta_0,g_0)\right)=o_p(n^{-1/2})
\end{align}
by stochastic equicontinuity.

In summary, we have
$$G_n(\hat\theta,\hat g)=G_n(\theta_0,g_0)+o_p(\|\hat\theta-\theta_0\|_2+n^{-1/2}),$$
and thus
$$A(\theta_0,g_0)\cdot (\hat\theta-\theta_0)=G_n(\theta_0,g_0)+o_p(\|\hat\theta-\theta_0\|_2+n^{-1/2}).$$

That is, we have
$$(\hat\theta-\theta_0)=A(\theta_0,g_0)^{-1}G_n(\theta_0,g_0)+o_p(\|\hat\theta-\theta_0\|_2+n^{-1/2}).$$

Since $G_n(\theta_0,g_0)$ is a mean-zero empirical process, we have that $\|G_n(\theta_0,g_0)\|_{2} = O_p(n^{-1/2}).$

By consistency of $\hat\theta$, above implies that
$$\hat\theta-\theta_0=O_p(n^{-1/2}).$$

Thus we get
$$(\hat\theta-\theta_0)=A(\theta_0,g_0)^{-1}G_n(\theta_0,g_0)+o_p(n^{-1/2}).$$

By Slutsky's Theorem, we conclude that
$$\sqrt{n}(\hat\theta-\theta_0)\overset{d}{\to}N\left(0, A(\theta_0,g_0)^{-1}
\mathbb{E}[m(Z;\theta_0,g_0)m(Z;\theta_0,g_0)']
A(\theta_0,g_0)^{-1}\right).$$
\end{proof}

We can also directly extend Lemma~2 to general nonlinear moments.
\begin{lemma}[Non-Linear Main Lemma]\label{lem:stoc-eq-jac-nonlin}
If the estimation algorithm satisfies the stability conditions: for all $i,j\in [p]$
\begin{align}
\max_{l\in[n]}\Big\|a_{i,j}(Z_l;\theta_0, \hat g)-a_{i,j}(Z_l;\ \theta_0, \hat g^{(-l)})\Big\|_1=~& o(n^{-1/2})\\
\max_{l\in[n]}\Big\|a_{i,j}(Z;\theta_0, \hat g)-a_{i,j}(Z;\theta_0, \hat g^{(-l)})\Big\|_2=~& o(n^{-1/2})\\
\max_{l\in[n]}\Big\|m_{i}(Z_l; \theta_0, \hat g)-m_{i}(Z_l;\theta_0, \hat g^{(-l)})\Big\|_1=~& o(n^{-1/2})\\
\max_{l\in[n]}\Big\|m_{i}(Z;\theta_0, \hat g)-m_{i}(Z; \theta_0, g^{(-l)})\Big\|_2=~& o(n^{-1/2})
\end{align}
and the moment satisfies the mean-squared-continuity condition:
\begin{align}
    \forall g,g': \E[(a_{i,j}(Z;\theta_0, g) - a_{i,j}(Z;\theta_0, g'))^2]\leq~& L \|g-g'\|_2^q\\
    \forall g,g': \E[(m_i(Z;\theta_0, g) - m_i(Z;\theta_0, g'))^2]\leq~& L \|g-g'\|_2^q
\end{align}
for some $0<q<\infty$ and some $L>0$, then the Condition~\eqref{eqn:equicont-nonlin} is satisfied.
\end{lemma}
\begin{proof}
The proof follows by replacing all functions $a(z,g), \nu(z,g)$ in the proof of Lemma~\ref{lem:stoc-eq-jac} correspondingly with the functions $a(z;\theta_0,g)$ and $m(z;\theta_0, g)$.
\end{proof}

\paragraph{Application to bagging estimators.} We finally note that if the moment satisfies Lipschitz conditions of the form:
\begin{align}
\E\left[\left(a_{i,j}(Z_l;\theta_0, \hat g) - a_{i,j}(Z_l;\theta_0, \hat{g}^{(-l)})\right)^{2}\right]\leq~ L\cdot \mathbb{E}\Big[\sup_z\|\hat g(z)-\hat g^{(-l)}(z)\|_2^{2r}\Big]^{1/r} \\
\max_{\theta\in \Theta}\, \E\left[\left(m_i(Z_l;\theta, \hat g) - m_{i}(Z_l;\theta, \hat{g}^{(-l)})\right)^{2}\right]\leq~ L\cdot \mathbb{E}\Big[\sup_z\|\hat g(z)-\hat g^{(-l)}(z)\|_2^{2r}\Big]^{1/r}
\end{align}
Then Theorem~\ref{bag:sleepy} can be applied to upper bound the right hand side of these inequalities by $o_p(n^{-1/2})$ for bagging estimators. This would then imply the stability conditions invoked in both the consistency and the normality theorem. Thus the main application for bagging estimators carries over to non-linear moments.

\appendix

\end{document}

%% file: dylan-macros.tex
\DeclarePairedDelimiter{\abs}{\lvert}{\rvert} %

\DeclareMathOperator*{\argmin}{arg\,min} 

\def\ddefloop#1{\ifx\ddefloop#1\else\ddef{#1}\expandafter\ddefloop\fi}
\def\ddef#1{\expandafter\def\csname bb#1\endcsname{\ensuremath{\mathbb{#1}}}}
\ddefloop ABCDEFGHIJKLMNOPQRSTUVWXYZ\ddefloop
\def\ddefloop#1{\ifx\ddefloop#1\else\ddef{#1}\expandafter\ddefloop\fi}
\def\ddef#1{\expandafter\def\csname b#1\endcsname{\ensuremath{\mathbf{#1}}}}
\ddefloop ABCDEFGHIJKLMNOPQRSTUVWXYZ\ddefloop
\def\ddef#1{\expandafter\def\csname c#1\endcsname{\ensuremath{\mathcal{#1}}}}
\ddefloop ABCDEFGHIJKLMNOPQRSTUVWXYZ\ddefloop
\def\ddef#1{\expandafter\def\csname h#1\endcsname{\ensuremath{\widehat{#1}}}}
\ddefloop ABCDEFGHIJKLMNOPQRSTUVWXYZ\ddefloop
\def\ddef#1{\expandafter\def\csname hc#1\endcsname{\ensuremath{\widehat{\mathcal{#1}}}}}
\ddefloop ABCDEFGHIJKLMNOPQRSTUVWXYZ\ddefloop
\def\ddef#1{\expandafter\def\csname t#1\endcsname{\ensuremath{\widetilde{#1}}}}
\ddefloop ABCDEFGHIJKLMNOPQRSTUVWXYZ\ddefloop
\def\ddef#1{\expandafter\def\csname tc#1\endcsname{\ensuremath{\widetilde{\mathcal{#1}}}}}
\ddefloop ABCDEFGHIJKLMNOPQRSTUVWXYZ\ddefloop

\usepackage{prettyref}

%% file: experiments.tex
\section{Experimental Evaluation}

We consider a synthetic experimental evaluation of our main theoretical findings. We focus on the partially linear model with a scalar outcome $Y \in \R$, a scalar continuous treatment $T\in \R$ and many controls $X\in \R^{n_x}$, where:
\begin{align}
    T =~& p_0(X) + \eta, & \eta \sim N(0, 1)\\
    Y =~& \theta_0 T + f_0(X) + \epsilon, & \epsilon\sim N(0,1)
\end{align}
Our goal is the estimation of the treatment effect $\theta_0$, while estimating the nuisance functions $p_0(X)$ and $q_0(X)=\theta_0 p_0(X) + f_0(X)$ in a flexible manner. We will consider estimation based on the orthogonal moment presented in Example~\ref{sleep2}. 

We considered sub-sampled $1$-nearest neighbor regression (NN) and sub-sampled fully grown (with only $1$ sample minimum leaf size) random forest (RF) regression for the estimation of $p_0$ and $q_0$, regressing $T$ on $X$ and $Y$ on $X$ correspondingly. 

The true functions $p_0$ and $f_0$ are actually linear in our data generating process, with $p_0(X)=\beta_0'X$ and $f_0(X)=\gamma_0'X$ and where $\beta_0$ and $\gamma_0$ have only one non-zero coefficient ($1$-sparse), and which is the same coefficient for both $\beta_0$ and $\gamma_0$. In other words, only one of the $n_x$ potential confounding variables $X$ is actually a confounder. 

We evaluate the performance of the estimate for $\theta_0$ for a range of values of the sample size $n$ and the dimension of the controls $n_x$ and with or without cross-fitting. For the cross-fitted estimates we used $2$ splits. For each specification we draw $1000$ experimental samples to evaluate the distributional properties of the estimate. 

We considered sub-sample sizes for the nuisance regressions based on our theoretical $n^{0.49}$ specification and for larger specifications too. We find that the estimate without cross-fitting typically has almost equal bias and smaller (and always comparable) variance (due most probably to the smaller mean squared error of the nuisances, since they are trained on larger sizes), especially in smaller sample sizes, and has better coverage properties, when the sub-sample size is $m=o(n^{1/2})$. Moreover, the estimate is approximately normally distributed, even without cross-fitting as is verified qualitatively via quantile-quantile (Q-Q) plots. 

We also report the mean of the estimate of the standard error across the $1000$ experiments, to evaluate the bias in the estimation of the standard error. We find that the estimate of the standard error is more accurate without cross-fitting, potentially because the estimate of the standard error does not incorporate the extra variance that stems from the sample-splitting process. This inaccuracy of the standard error estimate is most probably the reason for the worst coverage properties of the confidence intervals with cross-fitting. 

In summary, we verify experimentally that for stable estimators, with the theoretically required level of stability, sample splitting or cross-fitting is not needed to maintain asymptotic normality, small bias and nominal coverage.

As expected however, we do verify that as the sub-sample size becomes $m=O(n)$, where the estimate is both not stable and also does not optimize over hypothesis spaces with small critical radius or that satisfy the Donsker property, then the estimate without cross-fitting breaks down, while the estimate with cross-fitting maintains a decent performance, despite the high-variance of the nuisance estimate.

The experiments were run on a normal PC laptop with Processor Intel(R) Core(TM) i7-8650U CPU @ 1.90GHz, 2112 Mhz, 4 Core(s), 8 Logical Processor(s), 16GB RAM. It took around 1.5 hours to run all the experiments.

\newpage

\begin{figure}[H]
\scriptsize
\centering
\begin{subfigure}{0.48\textwidth}
\centering
\begin{tabular}{llrrrr}
\toprule
              &      &  bias &   std &  std\_est &  cov95 \\
\midrule
$n$=50, $n_x$=1 & cv=1 & 0.021 & 0.151 &    0.140 &  0.917 \\
              & cv=2 & 0.022 & 0.173 &    0.139 &  0.866 \\
$n$=50, $n_x$=2 & cv=1 & 0.042 & 0.149 &    0.143 &  0.920 \\
              & cv=2 & 0.046 & 0.170 &    0.142 &  0.863 \\
$n$=100, $n_x$=1 & cv=1 & 0.012 & 0.104 &    0.100 &  0.931 \\
              & cv=2 & 0.014 & 0.119 &    0.100 &  0.884 \\
$n$=100, $n_x$=2 & cv=1 & 0.030 & 0.110 &    0.101 &  0.911 \\
              & cv=2 & 0.032 & 0.121 &    0.101 &  0.878 \\
$n$=500, $n_x$=1 & cv=1 & 0.007 & 0.046 &    0.045 &  0.943 \\
              & cv=2 & 0.007 & 0.049 &    0.045 &  0.920 \\
$n$=500, $n_x$=2 & cv=1 & 0.015 & 0.045 &    0.045 &  0.927 \\
              & cv=2 & 0.015 & 0.048 &    0.045 &  0.917 \\
$n$=1000, $n_x$=1 & cv=1 & 0.003 & 0.032 &    0.032 &  0.946 \\
              & cv=2 & 0.002 & 0.034 &    0.032 &  0.922 \\
$n$=1000, $n_x$=2 & cv=1 & 0.010 & 0.031 &    0.032 &  0.944 \\
              & cv=2 & 0.010 & 0.033 &    0.032 &  0.930 \\
\bottomrule
\end{tabular}
\caption{Sub-sampled $1$-NN with $m=n^{0.49}$}
\end{subfigure}
\ \ \ \ \ \ 
\begin{subfigure}{0.48\textwidth}
\begin{tabular}{llrrrr}
\toprule
              &      &  bias &   std &  std\_est &  cov95 \\
\midrule
$n$=50, $n_x$=1 & cv=1 & 0.008 & 0.186 &    0.135 &  0.826 \\
              & cv=2 & 0.011 & 0.196 &    0.136 &  0.811 \\
$n$=50, $n_x$=2 & cv=1 & 0.011 & 0.193 &    0.139 &  0.834 \\
              & cv=2 & 0.020 & 0.196 &    0.138 &  0.831 \\
$n$=100, $n_x$=1 & cv=1 & 0.001 & 0.132 &    0.098 &  0.838 \\
              & cv=2 & 0.004 & 0.143 &    0.098 &  0.812 \\
$n$=100, $n_x$=2 & cv=1 & 0.003 & 0.135 &    0.098 &  0.824 \\
              & cv=2 & 0.011 & 0.141 &    0.099 &  0.828 \\
$n$=500, $n_x$=1 & cv=1 & 0.002 & 0.056 &    0.045 &  0.881 \\
              & cv=2 & 0.002 & 0.062 &    0.045 &  0.829 \\
$n$=500, $n_x$=2 & cv=1 & 0.003 & 0.056 &    0.045 &  0.878 \\
              & cv=2 & 0.003 & 0.060 &    0.045 &  0.851 \\
$n$=1000, $n_x$=1 & cv=1 & 0.001 & 0.039 &    0.032 &  0.883 \\
              & cv=2 & 0.000 & 0.043 &    0.032 &  0.849 \\
$n$=1000, $n_x$=2 & cv=1 & 0.001 & 0.039 &    0.032 &  0.891 \\
              & cv=2 & 0.000 & 0.043 &    0.032 &  0.853 \\
\bottomrule
\end{tabular}
\caption{Sub-sampled $1$-NN with $m=n^{10/11}$}
\end{subfigure}

\ \\
\ \\

\begin{subfigure}{0.48\textwidth}
\begin{tabular}{llrrrr}
\toprule
              &      &  bias &   std &  std\_est &  cov95 \\
\midrule
$n$=50, $n_x$=1 & cv=1 & 0.500 & 0.000 &    0.000 &  0.000 \\
              & cv=2 & 0.010 & 0.200 &    0.135 &  0.805 \\
$n$=50, $n_x$=2 & cv=1 & 0.500 & 0.000 &    0.000 &  0.000 \\
              & cv=2 & 0.017 & 0.204 &    0.138 &  0.795 \\
$n$=100, $n_x$=1 & cv=1 & 0.500 & 0.000 &    0.000 &  0.000 \\
              & cv=2 & 0.003 & 0.148 &    0.097 &  0.794 \\
$n$=100, $n_x$=2 & cv=1 & 0.500 & 0.000 &    0.000 &  0.000 \\
              & cv=2 & 0.008 & 0.146 &    0.098 &  0.797 \\
$n$=500, $n_x$=1 & cv=1 & 0.500 & 0.000 &    0.000 &  0.000 \\
              & cv=2 & 0.002 & 0.066 &    0.045 &  0.798 \\
$n$=500, $n_x$=2 & cv=1 & 0.500 & 0.000 &    0.000 &  0.000 \\
              & cv=2 & 0.002 & 0.064 &    0.045 &  0.838 \\
$n$=1000, $n_x$=1 & cv=1 & 0.500 & 0.000 &    0.000 &  0.000 \\
              & cv=2 & 0.000 & 0.046 &    0.031 &  0.817 \\
$n$=1000, $n_x$=2 & cv=1 & 0.500 & 0.000 &    0.000 &  0.000 \\
              & cv=2 & 0.001 & 0.045 &    0.032 &  0.829 \\
\bottomrule
\end{tabular}
\caption{Sub-sampled $1$-NN with $m=n$}
\end{subfigure}
\caption{Comparison of bias, variance and coverage properties, with (cv=2) and without (cv=1) cross-fitting (sample splitting), for the estimation of the treatment effect in the partially linear model, when a sub-sampled $1$-NN estimation is used for the nuisance function estimation. $n$ is the number of samples and $n_x$ the number of controls.}
\end{figure}

\begin{figure}[H]
\scriptsize
\centering
\begin{subfigure}{0.48\textwidth}
\centering
\begin{tabular}{llrrrr}
\toprule
               &      &  bias &   std &  std\_est &  cov95 \\
\midrule
$n$=50, $n_x$=5 & cv=1 & 0.102 & 0.149 &    0.144 &  0.873 \\
               & cv=2 & 0.103 & 0.165 &    0.143 &  0.836 \\
$n$=50, $n_x$=10 & cv=1 & 0.098 & 0.136 &    0.134 &  0.861 \\
               & cv=2 & 0.099 & 0.148 &    0.133 &  0.846 \\
$n$=100, $n_x$=5 & cv=1 & 0.066 & 0.102 &    0.101 &  0.894 \\
               & cv=2 & 0.064 & 0.109 &    0.101 &  0.877 \\
$n$=100, $n_x$=10 & cv=1 & 0.074 & 0.099 &    0.097 &  0.873 \\
               & cv=2 & 0.072 & 0.106 &    0.097 &  0.846 \\
$n$=500, $n_x$=5 & cv=1 & 0.020 & 0.044 &    0.045 &  0.930 \\
               & cv=2 & 0.021 & 0.046 &    0.045 &  0.919 \\
$n$=500, $n_x$=10 & cv=1 & 0.027 & 0.045 &    0.044 &  0.908 \\
               & cv=2 & 0.026 & 0.046 &    0.044 &  0.896 \\
$n$=1000, $n_x$=5 & cv=1 & 0.012 & 0.031 &    0.032 &  0.923 \\
               & cv=2 & 0.013 & 0.032 &    0.032 &  0.925 \\
$n$=1000, $n_x$=10 & cv=1 & 0.015 & 0.033 &    0.032 &  0.910 \\
               & cv=2 & 0.016 & 0.034 &    0.032 &  0.896 \\
\bottomrule
\end{tabular}
\caption{Sub-sampled Random Forest with $m=n^{0.49}$}
\end{subfigure}
\ \ \ \ \ \ 
\begin{subfigure}{0.48\textwidth}
\begin{tabular}{llrrrr}
\toprule
               &      &  bias &   std &  std\_est &  cov95 \\
\midrule
$n$=50, $n_x$=5 & cv=1 & 0.022 & 0.176 &    0.148 &  0.896 \\
               & cv=2 & 0.011 & 0.191 &    0.144 &  0.845 \\
$n$=50, $n_x$=10 & cv=1 & 0.026 & 0.167 &    0.145 &  0.904 \\
               & cv=2 & 0.002 & 0.180 &    0.140 &  0.869 \\
$n$=100, $n_x$=5 & cv=1 & 0.016 & 0.116 &    0.103 &  0.909 \\
               & cv=2 & 0.013 & 0.125 &    0.101 &  0.877 \\
$n$=100, $n_x$=10 & cv=1 & 0.018 & 0.114 &    0.103 &  0.908 \\
               & cv=2 & 0.016 & 0.128 &    0.100 &  0.870 \\
$n$=500, $n_x$=5 & cv=1 & 0.012 & 0.050 &    0.046 &  0.908 \\
               & cv=2 & 0.008 & 0.053 &    0.045 &  0.900 \\
$n$=500, $n_x$=10 & cv=1 & 0.013 & 0.048 &    0.046 &  0.923 \\
               & cv=2 & 0.011 & 0.051 &    0.045 &  0.911 \\
$n$=1000, $n_x$=5 & cv=1 & 0.010 & 0.035 &    0.033 &  0.924 \\
               & cv=2 & 0.006 & 0.036 &    0.032 &  0.912 \\
$n$=1000, $n_x$=10 & cv=1 & 0.012 & 0.035 &    0.033 &  0.909 \\
               & cv=2 & 0.008 & 0.038 &    0.032 &  0.893 \\
\bottomrule
\end{tabular}

\caption{Sub-sampled Random Forest with $m=n^{10/11}$}
\end{subfigure}
\caption{Comparison of bias, variance and coverage properties, with (cv=2) and without (cv=1) cross-fitting (sample splitting), for the estimation of the treatment effect in the partially linear model, when a sub-sampled Random Forest estimation is used for the nuisance function estimation. $n$ is the number of samples and $n_x$ the number of controls.}
\end{figure}

\begin{figure}[H]
\centering
\includegraphics[scale=.4]{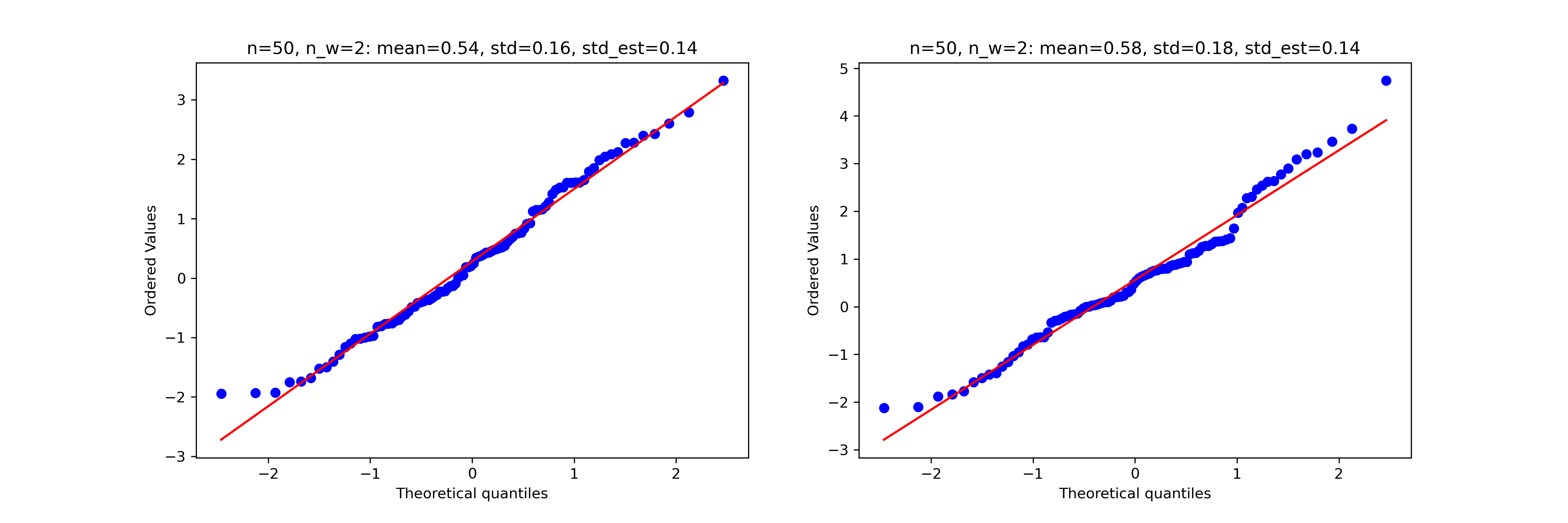}
\caption{Quantile-Quantile (Q-Q) plots for sub-sampled $1$-NN, with (right) and without (left) cross-fitting.}
\end{figure}

\begin{figure}[H]
\centering
\includegraphics[scale=.4]{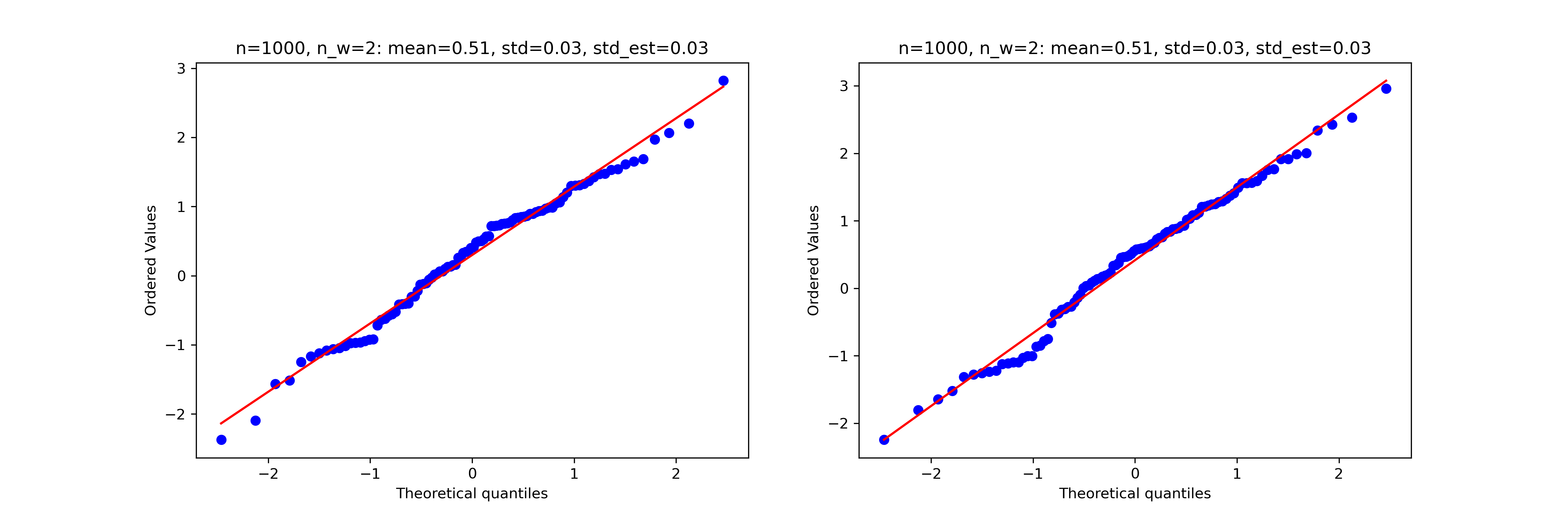}
\caption{Quantile-Quantile (Q-Q) plots for sub-sampled $1$-NN, with (right) and without (left) cross-fitting.}
\end{figure}

\begin{figure}[H]
\centering
\includegraphics[scale=.4]{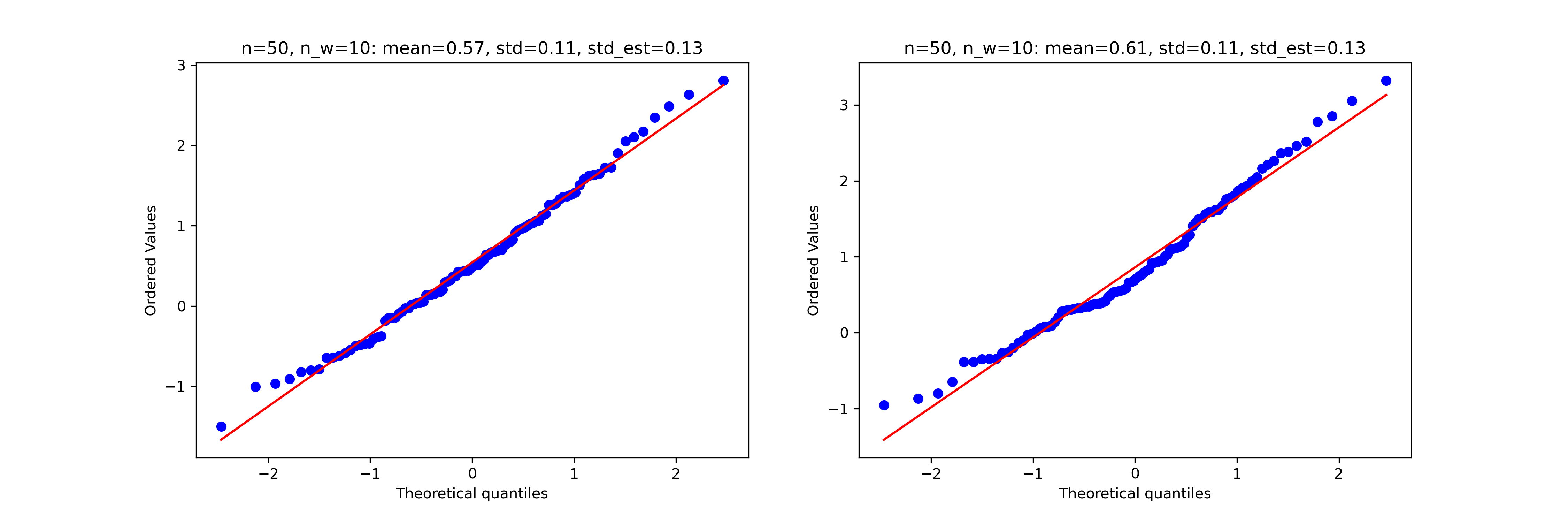}
\caption{Quantile-Quantile (Q-Q) plots for sub-sampled Random Forest, with (right) and without (left) cross-fitting.}
\end{figure}

\begin{figure}[H]
\centering
\includegraphics[scale=.4]{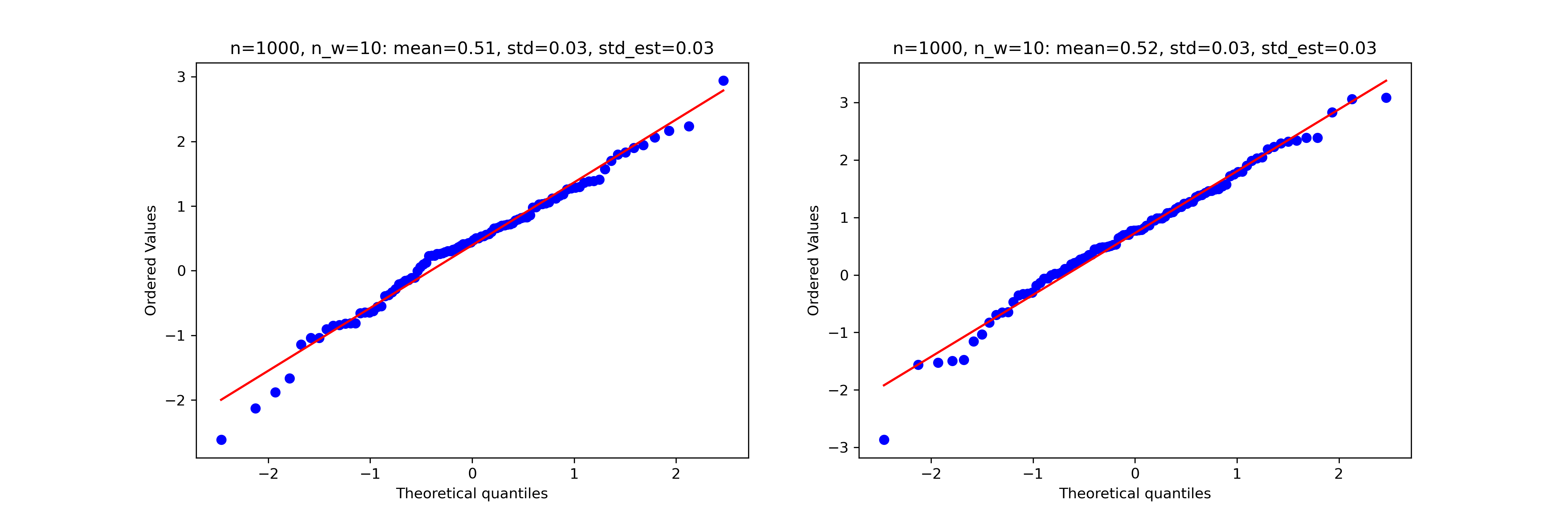}
\caption{Quantile-Quantile (Q-Q) plots for sub-sampled Random Forest, with (right) and without (left) cross-fitting.}
\end{figure}

%% file: general_m.bbl
\begin{thebibliography}{10}

\bibitem{ai2003efficient}
Chunrong Ai and Xiaohong Chen.
\newblock Efficient estimation of models with conditional moment restrictions
  containing unknown functions.
\newblock {\em Econometrica}, 71(6):1795--1843, 2003.

\bibitem{ai2007estimation}
Chunrong Ai and Xiaohong Chen.
\newblock Estimation of possibly misspecified semiparametric conditional moment
  restriction models with different conditioning variables.
\newblock {\em Journal of Econometrics}, 141(1):5--43, 2007.

\bibitem{ai2012semiparametric}
Chunrong Ai and Xiaohong Chen.
\newblock The semiparametric efficiency bound for models of sequential moment
  restrictions containing unknown functions.
\newblock {\em Journal of Econometrics}, 170(2):442--457, 2012.

\bibitem{austern2020asymptotics}
Morgane Austern and Wenda Zhou.
\newblock Asymptotics of cross-validation.
\newblock {\em arXiv preprint arXiv:2001.11111}, 2020.

\bibitem{bayle2020cross}
Pierre Bayle, Alexandre Bayle, Lucas Janson, and Lester Mackey.
\newblock Cross-validation confidence intervals for test error.
\newblock {\em Advances in Neural Information Processing Systems},
  33:16339--16350, 2020.

\bibitem{belloni2017program}
Alexandre Belloni, Victor Chernozhukov, Ivan Fern{\'a}ndez-Val, and Christian
  Hansen.
\newblock Program evaluation and causal inference with high-dimensional data.
\newblock {\em Econometrica}, 85(1):233--298, 2017.

\bibitem{belloni2011inference}
Alexandre Belloni, Victor Chernozhukov, and Christian Hansen.
\newblock Inference for high-dimensional sparse econometric models.
\newblock {\em arXiv:1201.0220}, 2011.

\bibitem{belloni2014inference}
Alexandre Belloni, Victor Chernozhukov, and Christian Hansen.
\newblock Inference on treatment effects after selection among high-dimensional
  controls.
\newblock {\em The Review of Economic Studies}, 81(2):608--650, 2014.

\bibitem{belloni2014uniform}
Alexandre Belloni, Victor Chernozhukov, and Kengo Kato.
\newblock Uniform post-selection inference for least absolute deviation
  regression and other {Z}-estimation problems.
\newblock {\em Biometrika}, 102(1):77--94, 2014.

\bibitem{belloni2014pivotal}
Alexandre Belloni, Victor Chernozhukov, and Lie Wang.
\newblock Pivotal estimation via square-root lasso in nonparametric regression.
\newblock {\em The Annals of Statistics}, 42(2):757--788, 2014.

\bibitem{bickel1982adaptive}
Peter~J Bickel.
\newblock On adaptive estimation.
\newblock {\em The Annals of Statistics}, pages 647--671, 1982.

\bibitem{bickel1993efficient}
Peter~J Bickel, Chris~AJ Klaassen, Ya'acov Ritov, and Jon~A Wellner.
\newblock {\em Efficient and adaptive estimation for semiparametric models},
  volume~4.
\newblock Johns Hopkins University Press, 1993.

\bibitem{bickel1988estimating}
Peter~J Bickel and Yaacov Ritov.
\newblock Estimating integrated squared density derivatives: Sharp best order
  of convergence estimates.
\newblock {\em Sankhy{\=a}: The Indian Journal of Statistics, Series A}, pages
  381--393, 1988.

\bibitem{boucheron2003concentration}
St{\'e}phane Boucheron, G{\'a}bor Lugosi, and Olivier Bousquet.
\newblock Concentration inequalities.
\newblock In {\em Summer school on machine learning}, pages 208--240. Springer,
  2003.

\bibitem{bousquet2002stability}
Olivier Bousquet and Andr{\'e} Elisseeff.
\newblock Stability and generalization.
\newblock {\em The Journal of Machine Learning Research}, 2:499--526, 2002.

\bibitem{bradic2017uniform}
Jelena Bradic and Mladen Kolar.
\newblock Uniform inference for high-dimensional quantile regression: Linear
  functionals and regression rank scores.
\newblock {\em arXiv:1702.06209}, 2017.

\bibitem{chernozhukov2018double}
Victor Chernozhukov, Denis Chetverikov, Mert Demirer, Esther Duflo, Christian
  Hansen, Whitney Newey, and James Robins.
\newblock Double/debiased machine learning for treatment and structural
  parameters: Double/debiased machine learning.
\newblock {\em The Econometrics Journal}, 21(1), 2018.

\bibitem{chernozhukov2016locally}
Victor Chernozhukov, Juan~Carlos Escanciano, Hidehiko Ichimura, Whitney~K
  Newey, and James~M Robins.
\newblock Locally robust semiparametric estimation.
\newblock {\em arXiv preprint arXiv:1608.00033}, 2016.

\bibitem{chernozhukov2015valid}
Victor Chernozhukov, Christian Hansen, and Martin Spindler.
\newblock Valid post-selection and post-regularization inference: An
  elementary, general approach.
\newblock {\em Annual Review of Economics}, 7(1):649--688, 2015.

\bibitem{chernozhukov2018global}
Victor Chernozhukov, Whitney Newey, and Rahul Singh.
\newblock Double/de-biased machine learning of global and local parameters
  using regularized {R}iesz representers.
\newblock {\em arXiv preprint arXiv:1802.08667}, 2018.

\bibitem{chernozhukov2020adversarial}
Victor Chernozhukov, Whitney Newey, Rahul Singh, and Vasilis Syrgkanis.
\newblock Adversarial estimation of riesz representers.
\newblock {\em arXiv preprint arXiv:2101.00009}, 2020.

\bibitem{chernozhukov2021simple}
Victor Chernozhukov, Whitney~K Newey, and Rahul Singh.
\newblock A simple and general debiased machine learning theorem with finite
  sample guarantees.
\newblock {\em arXiv preprint arXiv:2105.15197}, 2021.

\bibitem{elisseeff2003leave}
Andr{\'e} Elisseeff, Massimiliano Pontil, et~al.
\newblock Leave-one-out error and stability of learning algorithms with
  applications.
\newblock {\em NATO science series sub series iii computer and systems
  sciences}, 190:111--130, 2003.

\bibitem{hardt2016train}
Moritz Hardt, Ben Recht, and Yoram Singer.
\newblock Train faster, generalize better: Stability of stochastic gradient
  descent.
\newblock In {\em International conference on machine learning}, pages
  1225--1234. PMLR, 2016.

\bibitem{Hardt2016}
Moritz Hardt, Benjamin Recht, and Yoram Singer.
\newblock Train faster, generalize better: Stability of stochastic gradient
  descent.
\newblock In {\em Proceedings of the 33rd International Conference on
  International Conference on Machine Learning - Volume 48}, ICML'16, pages
  1225--1234. JMLR.org, 2016.

\bibitem{hasminskii1979nonparametric}
Rafail~Z Hasminskii and Ildar~A Ibragimov.
\newblock On the nonparametric estimation of functionals.
\newblock In {\em Proceedings of the Second Prague Symposium on Asymptotic
  Statistics}, 1979.

\bibitem{ibragimov1981statistical}
I~Ibragimov and R~Has'minskii.
\newblock Statistical estimation, vol. 16 of.
\newblock {\em Applications of Mathematics}, 1981.

\bibitem{ibragimov1998exact}
Rustam Ibragimov and Sh~Sharakhmetov.
\newblock On an exact constant for the rosenthal inequality.
\newblock {\em Theory of Probability \& Its Applications}, 42(2):294--302,
  1998.

\bibitem{jankova2015confidence}
Jana Jankova and Sara Van De~Geer.
\newblock Confidence intervals for high-dimensional inverse covariance
  estimation.
\newblock {\em Electronic Journal of Statistics}, 9(1):1205--1229, 2015.

\bibitem{jankova2016confidence}
Jana Jankova and Sara Van De~Geer.
\newblock Confidence regions for high-dimensional generalized linear models
  under sparsity.
\newblock {\em arXiv:1610.01353}, 2016.

\bibitem{jankova2018semiparametric}
Jana Jankova and Sara Van De~Geer.
\newblock Semiparametric efficiency bounds for high-dimensional models.
\newblock {\em The Annals of Statistics}, 46(5):2336--2359, 2018.

\bibitem{javanmard2014confidence}
Adel Javanmard and Andrea Montanari.
\newblock Confidence intervals and hypothesis testing for high-dimensional
  regression.
\newblock {\em The Journal of Machine Learning Research}, 15(1):2869--2909,
  2014.

\bibitem{javanmard2014hypothesis}
Adel Javanmard and Andrea Montanari.
\newblock Hypothesis testing in high-dimensional regression under the
  {G}aussian random design model: Asymptotic theory.
\newblock {\em IEEE Transactions on Information Theory}, 60(10):6522--6554,
  2014.

\bibitem{javanmard2018debiasing}
Adel Javanmard and Andrea Montanari.
\newblock Debiasing the lasso: Optimal sample size for {G}aussian designs.
\newblock {\em The Annals of Statistics}, 46(6A):2593--2622, 2018.

\bibitem{Kale11cross-validationand}
Satyen Kale, Ravi Kumar, and Sergei Vassilvitskii.
\newblock Cross-validation and mean-square stability.
\newblock In {\em In Proceedings of the Second Symposium on Innovations in
  Computer Science (ICS2011}, pages 487--495, 2011.

\bibitem{khosravi2019non}
Khashayar Khosravi, Greg Lewis, and Vasilis Syrgkanis.
\newblock Non-parametric inference adaptive to intrinsic dimension.
\newblock {\em arXiv preprint arXiv:1901.03719}, 2019.

\bibitem{klaassen1987consistent}
Chris~AJ Klaassen.
\newblock Consistent estimation of the influence function of locally
  asymptotically linear estimators.
\newblock {\em The Annals of Statistics}, pages 1548--1562, 1987.

\bibitem{kosorok2007introduction}
Michael~R Kosorok.
\newblock {\em Introduction to empirical processes and semiparametric
  inference}.
\newblock Springer Science \& Business Media, 2007.

\bibitem{levit1976efficiency}
B~Ya Levit.
\newblock On the efficiency of a class of non-parametric estimates.
\newblock {\em Theory of Probability \& Its Applications}, 20(4):723--740,
  1976.

\bibitem{luedtke2016statistical}
Alexander~R Luedtke and Mark~J Van Der~Laan.
\newblock Statistical inference for the mean outcome under a possibly
  non-unique optimal treatment strategy.
\newblock {\em Annals of Statistics}, 44(2):713, 2016.

\bibitem{newey1994asymptotic}
Whitney~K Newey.
\newblock The asymptotic variance of semiparametric estimators.
\newblock {\em Econometrica}, pages 1349--1382, 1994.

\bibitem{newey1998undersmoothing}
Whitney~K Newey, Fushing Hsieh, and James~M Robins.
\newblock Undersmoothing and bias corrected functional estimation.
\newblock Technical report, MIT Department of Economics, 1998.

\bibitem{newey2004twicing}
Whitney~K Newey, Fushing Hsieh, and James~M Robins.
\newblock Twicing kernels and a small bias property of semiparametric
  estimators.
\newblock {\em Econometrica}, 72(3):947--962, 2004.

\bibitem{neykov2018unified}
Matey Neykov, Yang Ning, Jun~S Liu, and Han Liu.
\newblock A unified theory of confidence regions and testing for
  high-dimensional estimating equations.
\newblock {\em Statistical Science}, 33(3):427--443, 2018.

\bibitem{neyman1959}
Jerzy Neyman.
\newblock Optimal asymptotic tests of composite statistical hypotheses.
\newblock In {\em Probability and Statistics}, pages 416--444. Wiley, 1959.

\bibitem{neyman1979c}
Jerzy Neyman.
\newblock C ($\alpha$) tests and their use.
\newblock {\em Sankhy{\=a}: The Indian Journal of Statistics, Series A}, pages
  1--21, 1979.

\bibitem{ning2017general}
Yang Ning and Han Liu.
\newblock A general theory of hypothesis tests and confidence regions for
  sparse high dimensional models.
\newblock {\em The Annals of Statistics}, 45(1):158--195, 2017.

\bibitem{pfanzagl1982lecture}
Johann Pfanzagl.
\newblock Lecture notes in statistics.
\newblock {\em Contributions to a general asymptotic statistical theory}, 13,
  1982.

\bibitem{ren2015asymptotic}
Zhao Ren, Tingni Sun, Cun-Hui Zhang, and Harrison~H Zhou.
\newblock Asymptotic normality and optimalities in estimation of large
  {G}aussian graphical models.
\newblock {\em The Annals of Statistics}, 43(3):991--1026, 2015.

\bibitem{rio2009moment}
Emmanuel Rio.
\newblock Moment inequalities for sums of dependent random variables under
  projective conditions.
\newblock {\em Journal of Theoretical Probability}, 22(1):146--163, 2009.

\bibitem{robins1995semiparametric}
James~M Robins and Andrea Rotnitzky.
\newblock Semiparametric efficiency in multivariate regression models with
  missing data.
\newblock {\em Journal of the American Statistical Association},
  90(429):122--129, 1995.

\bibitem{robins1995analysis}
James~M Robins, Andrea Rotnitzky, and Lue~Ping Zhao.
\newblock Analysis of semiparametric regression models for repeated outcomes in
  the presence of missing data.
\newblock {\em Journal of the American Statistical Association},
  90(429):106--121, 1995.

\bibitem{robinson1988root}
Peter~M Robinson.
\newblock Root-n-consistent semiparametric regression.
\newblock {\em Econometrica: Journal of the Econometric Society}, pages
  931--954, 1988.

\bibitem{toth2016tmle}
B~Toth and MJ~van~der Laan.
\newblock {TMLE} for marginal structural models based on an instrument.
\newblock Technical report, UC Berkeley Division of Biostatistics, 2016.

\bibitem{tsiatis2007semiparametric}
Anastasios Tsiatis.
\newblock {\em Semiparametric theory and missing data}.
\newblock Springer Science \& Business Media, 2007.

\bibitem{van2014asymptotically}
Sara Van~de Geer, Peter B{\"u}hlmann, Ya'acov Ritov, and Ruben Dezeure.
\newblock On asymptotically optimal confidence regions and tests for
  high-dimensional models.
\newblock {\em The Annals of Statistics}, 42(3):1166--1202, 2014.

\bibitem{van2011targeted}
Mark~J Van~der Laan and Sherri Rose.
\newblock {\em Targeted Learning: Causal Inference for Observational and
  Experimental Data}.
\newblock Springer Science \& Business Media, 2011.

\bibitem{van2006targeted}
Mark~J Van Der~Laan and Daniel Rubin.
\newblock Targeted maximum likelihood learning.
\newblock {\em The International Journal of Biostatistics}, 2(1), 2006.

\bibitem{van1991differentiable}
Aad Van Der~Vaart et~al.
\newblock On differentiable functionals.
\newblock {\em The Annals of Statistics}, 19(1):178--204, 1991.

\bibitem{vaart}
Aad~W Van~der Vaart.
\newblock {\em Asymptotic Statistics}, volume~3.
\newblock Cambridge University Press, 2000.

\bibitem{zhang2014confidence}
Cun-Hui Zhang and Stephanie~S Zhang.
\newblock Confidence intervals for low dimensional parameters in high
  dimensional linear models.
\newblock {\em Journal of the Royal Statistical Society: Series B (Statistical
  Methodology)}, 76(1):217--242, 2014.

\bibitem{zheng2010asymptotic}
Wenjing Zheng and Mark~J Van Der~Laan.
\newblock Asymptotic theory for cross-validated targeted maximum likelihood
  estimation.
\newblock {\em U.C. Berkeley Division of Biostatistics Working Paper Series.
  Working Paper 273}, 2010.

\bibitem{zhu2017breaking}
Yinchu Zhu and Jelena Bradic.
\newblock Breaking the curse of dimensionality in regression.
\newblock {\em arXiv:1708.00430.}, 2017.

\bibitem{zhu2018linear}
Yinchu Zhu and Jelena Bradic.
\newblock Linear hypothesis testing in dense high-dimensional linear models.
\newblock {\em Journal of the American Statistical Association},
  113(524):1583--1600, 2018.

\end{thebibliography}
